\documentclass[review]{siamart251104}
\usepackage{graphicx}
\usepackage{pdfpages}
\usepackage{tabto}
\usepackage{transparent}
\usepackage{algorithm2e}
\usepackage{chngcntr}
\usepackage{xr}
\usepackage{floatrow}
\newfloatcommand{capbtabbox}{table}[][0.48\textwidth]

\usepackage{pdfpages} 
\usepackage{pgffor} 

\makeatletter
\AtBeginDocument{\let\LS@rot\@undefined}
\makeatother

\def\supplementfilename{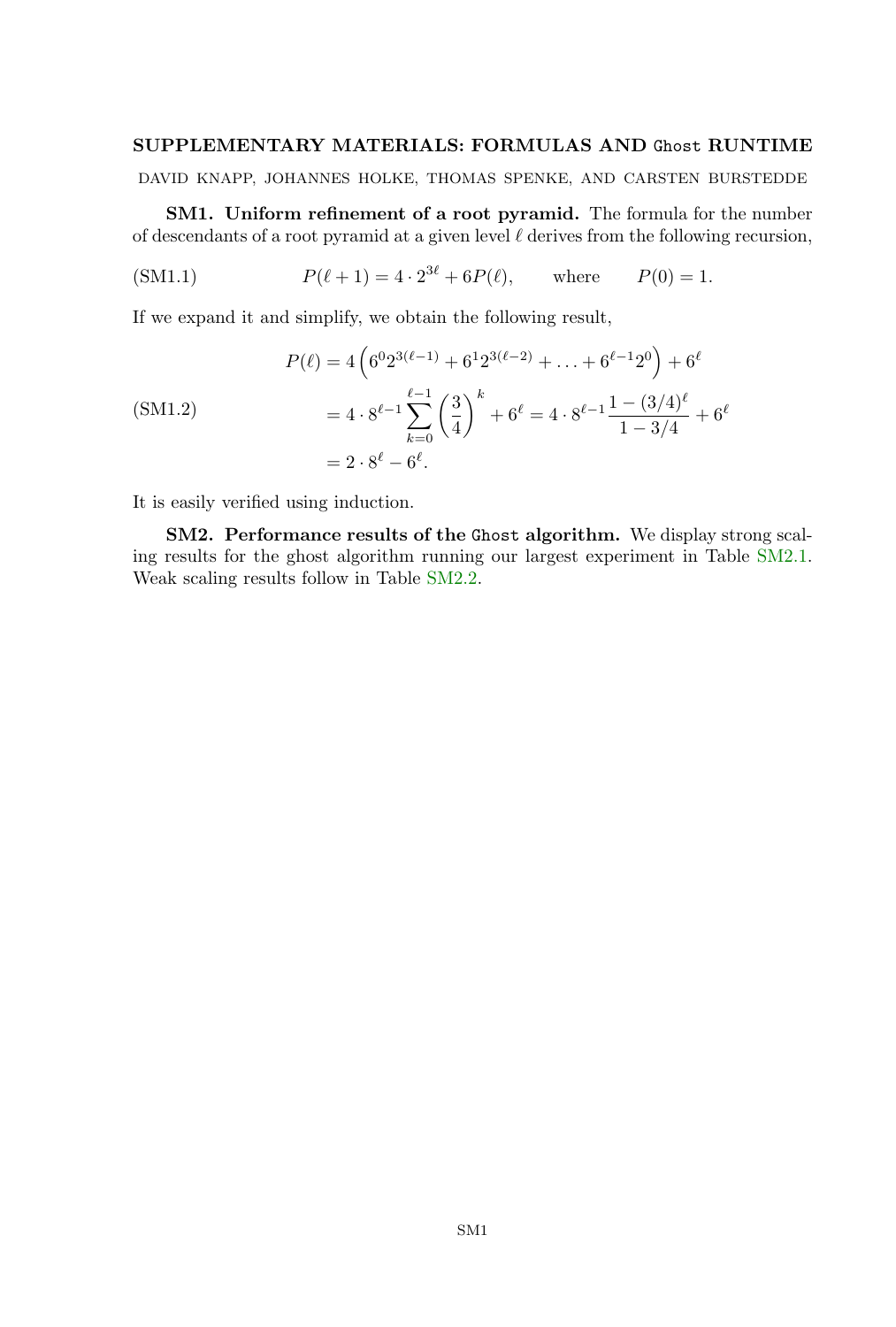}

\pdfximage{\supplementfilename}
\def\numbersupplementpages{\the\pdflastximagepages}

\newif\ifarXiv
\arXivtrue 

\usepackage{amsfonts}
\usepackage{graphicx}
\usepackage{epstopdf}
\usepackage{tikz}

\usepackage{amssymb}
\usepackage{booktabs}
\usepackage{cite} 
\usepackage[strict,style=italian]{csquotes} 
\usepackage[T1]{fontenc} 
\usepackage{mathrsfs} 
\usepackage{multirow}	  
\usepackage[dvipsnames]{xcolor} 	
\usepackage{enumerate}

\usepackage{subcaption} 
\usepackage{longtable}
\usepackage{tabularx}
\usepackage{verbatim}
\usepackage{xspace}
\usepackage{sidecap}

\newcommand{\TheTitle}{A Morton-type space-filling curve for\\
                       pyramid subdivision and hybrid adaptive mesh refinement}
\newcommand{\TheShortTitle}{A Morton-type SFC for pyramid subdivision and hybrid AMR}
\newcommand{\TheAuthors}{D.\ Knapp, J.\ Holke, T.\ Spenke, C.\ Burstedde and L. Dreyer}
\newcommand{\TheAuthorsHeader}{\TheAuthors}

\headers{\TheShortTitle}{\TheAuthorsHeader}

\title{{\TheTitle}
}
\author{
  David Knapp\thanks{Universität zu Köln, Germany,
  and Institute of Software Technology, German Aerospace Center (DLR), Cologne, Germany (\email{david.knapp@dlr.de})}
  \and
  Johannes Albrecht Holke%
  \thanks{Institute of Software Technology, DLR, Cologne, Germany
  (\email{johannes.holke@dlr.de})}
  \and
  Thomas Spenke%
  \thanks{Institute of Software Technology, DLR, Cologne, Germany
  (\email{thomas.spenke@dlr.de})}
  \and
  Carsten Burstedde\thanks{Institut f{\"u}r Numerische Simulation (INS)
    and Hausdorff Center for Mathematics (HCM),
    Rheinische Friedrich-Wilhelms-Universit{\"a}t Bonn, Germany}
  \and
  Lukas Dreyer\thanks{Institute of Software Technology, DLR, Cologne, Germany and Leibniz University Hannover, Institute of Applied Mathematics (IfAM), Leibniz University Hannover, Germany and Cluster of
Excellence PhoenixD (Photonics, Optics, and Engineering - Innovation Across Disciplines), Leibniz
University Hannover, Germany
}
}

\theoremstyle{plain}
\newtheorem{defi}[theorem]{Definition}

\makeatletter
\AtBeginDocument{
\def\refstepcounter@optarg[#1]#2{%
  \cref@old@refstepcounter{#2}%
  \cref@constructprefix{#2}{\cref@result}%
  \@ifundefined{cref@#1@alias}%
    {\def\@tempa{#1}}%
    {\def\@tempa{\csname cref@#1@alias\endcsname}}%
  \protected@edef\cref@currentlabel{%
    [\@tempa][\arabic{#2}][\cref@result]%
    \csname p@#2\endcsname\csname the#2\endcsname}}}      
\makeatother

\numberwithin{equation}{section}

\ifpdf
  \DeclareGraphicsExtensions{.eps,.pdf,.png,.jpg}
\else
  \DeclareGraphicsExtensions{.eps}
\fi

\newcommand{\braceset}[1]{\lbrace #1 \rbrace}

\newcommand{\bitwand}{\;\&\;} 

\newcommand{\pforest}{\texttt{p4est}\xspace} 
\newcommand{\tetcode}{\texttt{t8code}\xspace} 

\newcommand{\ghostn}{\texttt{Ghost}\xspace} 

\newcommand{\tetnew}{\texttt{New}\xspace}
\newcommand{\tetadapt}{\texttt{Adapt}\xspace}
\newcommand{\tetpartition}{\texttt{Partition}\xspace}

\newcommand{\tetghost}{\ghostn}

\newif\ifANMan
\ANMantrue 
\newif\ifDEBUG
\ifANMan
\newcommand{\ANM}[1]{\footnote{ANMERKUNG: #1}}  
\newcommand{\ANMtext}[1]{\footnotetext{ANMERKUNG: #1}}
\newcommand{\ANMmarkn}[1]{\footnotemark[#1]}           
\newcommand{\ANMtextn}[2]{\footnotetext[#1]{ANMERKUNG: #2}}
\else
\newcommand{\ANM}[1]{}  
\newcommand{\ANMtext}[1]{}
\newcommand{\ANMmarkn}[1]{}          
\newcommand{\ANMtextn}[2]{}
\fi

\ifDEBUG
\newcommand{\figlab}[1]{\caption*{\color{red} #1}\label{fig:#1}}
\else
\newcommand{\figlab}[1]{\label{fig:#1}}
\fi
\newcommand{\figref}[1]{Figure~\ref{fig:#1}}
\newcommand{\tablab}[1]{\label{tab:#1}}
\newcommand{\tabref}[1]{Table~\ref{tab:#1}}
\newcommand{\eqnlab}[1]{\label{Equ:#1}}
\newcommand{\eqnref}[1]{\eqref{Equ:#1}}
\newcommand{\seclab}[1]{\label{Sec:#1}}
\newcommand{\secref}[1]{Section~\ref{Sec:#1}}

\definecolor{mygreen}{rgb}{0,0.7,0}

\DeclareMathOperator{\ptype}{type}				

\newcommand{\plen}{\mathrm{len}}				
\newcommand{\pquest}{\;?\;}                             

\newcommand{\fC}{\mathfrak{C}}
\newcommand{\fK}{\mathfrak{K}}

\newcommand{\floors}[1]{\left\lfloor #1 \right\rfloor}

\newcolumntype{L}[1]{>{\raggedright\arraybackslash}p{#1}}
\newcolumntype{C}[1]{>{\centering\arraybackslash}p{#1}}
\newcolumntype{R}[1]{>{\raggedleft\arraybackslash}p{#1}}

\SetKwComment{Comment}{$\triangleright\,$}{}

\SetCommentSty{myCommentStyle}
\newlength{\algoCommentLength}
\newcommand{\mygets}[1]{=}

\newcommand{\myAlgoComment}[1]{\Comment*[r]{\makebox[\algoCommentLength]{#1\hfill}}}
\newcommand{\myAlgoCode}[1]{#1}
\newcommand{\myAlgoLine}[2]{\myAlgoCode{#2} \myAlgoComment{#1}}

\newcommand{\switchShapeLevel}{\text{min\_tet\_level}}

\newcommand{\level}{\text{level}}
\newcommand{\coord}{\text{coord}}
\newcommand{\type}{\text{type}}

\newcommand{\myAlgoBlockHeader}[1]{\textit{\textbf{Note:}} \textit{#1}}

\counterwithin{algocf}{section}

\newcommand{\interleaving}{\dot{\bot}}

\newcommand{\ourSubSubsection}[1]{\subsubsection*{#1}}

\newcommand{\enumerateCases}[1]{
\begin{enumerate}[(a)] #1 \end{enumerate}}
\newcommand{\enumerateSteps}[1]{
\begin{enumerate}[1.] #1 \end{enumerate}}

\ifpdf
\hypersetup{%
  pdftitle={\TheTitle},
  pdfauthor={\TheAuthors}%
}
\fi

\graphicspath{{pictures/pyramids/}
}%

\begin{document}

\maketitle

\begin{abstract}%
The forest-of-refinement-trees approach allows for dynamic adaptive mesh refinement (AMR) at negligible cost. While originally developed for quadrilateral and hexahedral elements, previous work established the theory and algorithms for unstructured meshes of simplicial and prismatic elements. To harness the full potential of tree-based AMR for three-dimensional mixed-element meshes, this paper introduces the pyramid as a new functional element type;
its primary purpose is to
connect tetrahedral and hexahedral elements without hanging edges.
We present a well-defined space-filling curve (SFC) for the pyramid and
detail how the unique challenges on the element and forest level associated
with the pyramidal refinement are resolved. We propose the necessary
functional design and generalize the fundamental global parallel algorithms
for refinement, coarsening, partitioning, and face ghost exchange to fully
support this new element. Our demonstrations confirm the efficiency and
scalability of this complete, hybrid-element dynamic AMR framework.
\end{abstract}




\section{Introduction}
\seclab{intro}
Numerical simulations are indispensable for analyzing complex physical systems.
The fidelity of their results depends critically on the spatial resolution
of the simulation domain. Adaptive meshes are widely used for efficient
resolution of local features, which reduces computational cost
\cite{FehlingBangerth, DavydovGerasimovPelteretSteinman,
ranocha2022adaptive, kirby2019wind}. The prevailing strategy for
distributing the computational mesh onto multiple processes is
graph-oriented partitioning \cite{GraphPartitioningShephard,
GraphPartitioningShephard2, WALSHAW1997102, GraphPartOpenFoam}. It typically
relies on third-party libraries \cite {KarypisKumar95,
KarypisSchloegelKumar, pellegrini:hal-00770422, Zoltan} to compute
near-optimal partitions that minimize inter-process communication. While
these algorithms are well established for static unstructured meshes, where
the partition is computed only once, their computational overhead becomes a
significant factor in dynamic scenarios and for extreme process counts.

A robust alternative to graph-based partitioning is dynamic, hierarchical
adaptation based on a forest of refinement trees. This approach exploits the
locality-preserving properties of \textit{space-filling curves} (SFCs)
\cite{Peano90, Hilbert91, Lebesgue04, Sierpinski12, Samet06} to linearize
the multidimensional domain. While classically restricted to quadrilateral
or hexahedral elements \cite{StewartEdwards04}, this method allows for
rapid, logic-based repartitioning without the overhead of global graph
algorithms. Consequently, it is highly effective for massively parallel,
dynamic simulations, as demonstrated by several scalable implementations
\cite{GriebelZumbusch99, TuOHallaronGhattas05,
BursteddeWilcoxGhattas11, WeinzierlMehl11}, some of which are now integrated
into general-purpose parallel-computing libraries
\cite{BangerthBursteddeHeisterEtAl11, IsaacKnepley15,
BursteddeKirilinKlofkorn25}.

Most of these libraries support the usage of a fixed element shape, such as
pure hexahedral or pure tetrahedral meshes. Both choices come with their own
strengths and weaknesses; for example, hexahedral elements are beneficial
to various classes of solvers, while tetrahedra offer more flexibility to
cover complex geometries.
Combining both shapes can be advantageous but requires bridging elements,
such as prisms or pyramids, to prevent hanging edges \cite{hybridMesh,
hybridMeshHangingEdge}.
By themselves, these shapes do not offer further advantages, but they allow
for conforming transitions wherever tetrahedral and hexahedral mesh elements
meet.

In this paper, we present a generalization of Morton-type space-filling
curves \cite{Morton66} to pyramids, a crucial development that enables fully
general, tree-based hybrid adaptive mesh refinement (AMR) \cite{hybridMesh,
ItoShihSoni}. In previous works, we have established the underlying theory,
parallel algorithms, and scalable demonstrations for hypercubical and
simplicial elements up to three dimensions \cite{BursteddeHolke16},
and we have conducted initial research on suitable space-filling
curves for prisms \cite{Knapp17} and pyramids \cite{Knapp20}.
To fully support the pyramid primitive in practice, this work introduces a
set of new per-element algorithms, such as computing the child, parent, and
face neighbor, along with extensions to global mesh modification and
interrogation algorithms like computing the bounds of a uniformly
partitioned mesh.  We present a reference implementation within the
open-source adaptive-meshing software library
\tetcode \cite{tetcodeweb23}.

We propose a novel refinement strategy that partitions a pyramid into six pyramidal and four tetrahedral children. The decisive advantage of this decomposition is its commutativity with cubic refinement. Analogous to the tetrahedral Morton index \cite{Holke18}, this property allows us to uniquely identify any pyramid in the forest solely by its enclosing sub-cube and a local orientation integer.
Unlike standard element shapes that refine self-similarly, this strategy results in heterogeneous children (pyramids and tetrahedra). While this geometric irregularity poses a significant challenge for low-level index arithmetic, our solution overcomes this barrier, providing the essential component to enable fully general, tree-based hybrid adaptive mesh refinement.
%


\subsection{Summary}
\label{Subsection}
Our proposed refinement strategy introduces heterogeneity in the
descendants, as children of a pyramid are either tetrahedra or pyramids.
Furthermore, the refinement pattern leads to a structure that is neither
strictly self-similar nor uniformly balanced, with child numbers of eight
(for tetrahedra) or ten (for pyramids).
This geometric complexity requires sophisticated, novel approaches for the
atomic element algorithms.
For example, generating a sibling element, which is a straightforward
operation for hypercubes and manageable for simplices, becomes non-trivial
due to the mixed nature of the children.
We propose functional designs for these essential algorithms and describe
how the fundamental global parallel routines for refinement, coarsening, and
gathering ghost elements are successfully generalized to robustly include
pyramids.

%
%
%
%
%

The design of the space-filling curve itself is achieved through a
relatively benign generalization, following the successful methodology
developed previously for tetrahedra.
Crucially, primitives like six tetrahedra, two prisms, or the combination of
two pyramids and two tetrahedra (as shown in \figref{pyra_types}) may be
arranged to tile a reference cube.
We complement the standard cubic Morton SFC by introducing a type index that
uniquely counts through the multiple, similar incarnations of a given
primitive within that reference cube.

The partitioning algorithm follows without additional efforts, since it
relies entirely on the SFC.
The adaptive refinement and coarsening, however, are complicated significantly
by several novel aspects:
\begin{enumerate}
\item
  Each pyramid refines into six pyramids and four tetrahedra; each of the
  pyramids may have one of two orientations in relation to the reference cube.
  This makes the generation of an initial uniform mesh non-trivial, since it
  contains a mix of equal-depth pyramidal and tetrahedral elements.
\item
  Formally generating the parent within pyramid refinement trees is quite
  involved compared to the much simpler algorithms for hexahedra,
  tetrahedra, and prisms, since it is not obvious whether the parent of a
  tetrahedron is a tetrahedron or a pyramid.
\item
  The face neighbor of a pyramid may be taken across a triangular or a
  quadrilateral face, as opposed to standard element shapes in \tetcode like
  tetrahedra and hexahedra, for which the face always has the same shape.
  This requires additional lookup tables and branching.
\end{enumerate}

\secref{tetcode} briefly reviews the forest-of-octrees structure implemented in \tetcode. We define and encode the space-filling curve for pyramids in \secref{pyramidSFC}. \secref{lowlevel} develops the necessary per-element algorithms. In \secref{highlevel}, we describe the global parallel algorithms for mesh manipulation, including the generation of the initial hybrid forest. We conclude in \secref{Sec:resultcube} and \secref{Sec:resultplane} with demonstrations of the efficiency and scalability of this new, fully hybrid approach to dynamic AMR.

\section{\tetcode and its parallel view of the forest}
\seclab{tetcode}

The open-source software
\tetcode \cite{tetcodeweb23} is our reference implementation of
parallel, hybrid, dynamic, adaptive mesh refinement via a suite of
dedicated parallel algorithms.
The word \textit{parallel} refers to the distributed storage of mesh data
and metadata as well as its modification via concurrent algorithms
that communicate over MPI \cite{Forum15}.
With \textit{hybrid} we refer to mixing elements of different primitive
shapes in the same mesh, and \textit{dynamic} refers to generating and
modifying a parallel mesh in-core during a simulation at negligible cost
compared to the overall runtime.
Presently, and not barring future extension, it encodes topological
neighbors across mesh faces.
Thus, \tetcode addresses a subset of the full face, edge, and
corner connectivity supported by the \pforest software \cite{Burstedde25a}.
\pforest, on the other hand, is focused strictly on quadrilateral and
hexahedral elements and offers no route towards the hybrid refinement
native to \tetcode.

\tetcode \cite{BursteddeHolke16}, like \pforest
\cite{BursteddeWilcoxGhattas11} and deal.II
\cite{BangerthHartmannKanschat07} before, uses an unstructured mesh as a
starting point for its operations, the so-called \textit{coarse mesh}, which
is usually given as the output of a mesh generator.
Every element in this coarse mesh corresponds to one tree in the overall
data structure, known as the \emph{forest} made of these trees.
The coarse mesh of \tetcode encodes the shape of its root elements and the
connection to a neighboring tree for each of its faces not on the domain
boundary.
Adaptively refining the trees, we create a distributed set of \textit{leaves}
that become the elements of the refined mesh.
The space-filling curve induces a consecutive index over the leaf elements,
which enables simple and efficient load balancing across multiple
processes.
Storing only the leaves of the tree, decidedly ignoring all interior tree
nodes, has been pioneered with the dendro code \cite{SundarSampathBiros08}
and is often called \emph{linear} octree storage.
This concept enables supremely efficient data handling and evolves from the
first generation of distributed octree codes \cite{TuOHallaronGhattas05}.


As illustrated in Figure \ref{fig:forest_of_tree}, each process holds its local elements and may also store its neighboring elements on other processes, which we call \textit{ghost} elements. In 3D, \tetcode currently supports hexahedral, tetrahedral and prism elements using space filling curves. However, a suitable space-filling curve for pyramids did not exist so far. We address this gap by introducing a novel curve in this work.

\begin{figure}
    \centering
    \includegraphics[width=0.8\linewidth]{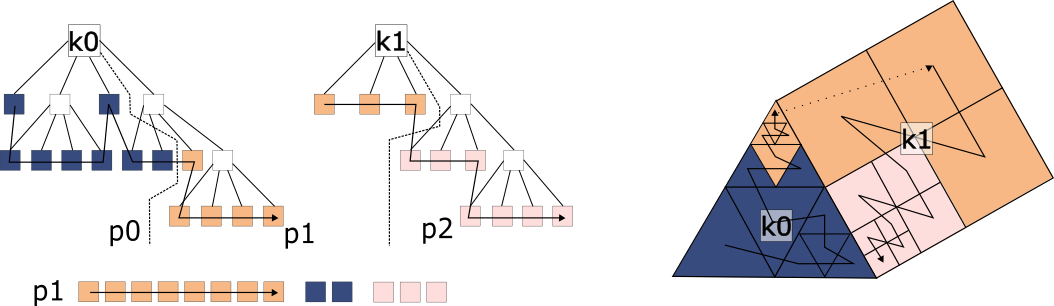}
    \caption{The forest-of-trees approach used in \tetcode. The trees $k_0$ and $k_1$ are uniformly partitioned over three processes $p_0$, $p_1$, $p_2$. Each process holds its local elements and the ghost elements. The black arrow indicates the order of the elements by the space-filling curve index. The array of elements shows the elements stored by $p_1$.}
    \label{fig:forest_of_tree}
\end{figure}

\section{A pyramidal refinement and space-filling curve}
\seclab{pyramidSFC}

In this section we introduce the mathematical concept of a hierarchical
pyramidal refinement and define the space-filling curve it induces.
The anchor of our construction is the root pyramid, for which we specify a
refinement scheme that defines its children. Following the construction of
children of other elements, i.e., cutting the corners of the element, we get
pyramidal children in the corners of a pyramid and tetrahedra in between;
see \figref{pyra_refinment_type_6} and \figref{pyra_refinment_type_7}.

Upon further refinement, the tetrahedra resulting from the pyramidal decomposition will refine self-similarly, meaning they yield only tetrahedral children. For this, we can directly reuse published concepts, definitions, and algorithms based on Bey's refinement scheme \cite{BursteddeHolke16, Bey92}. A reference cube is split into six similar tetrahedra such that any descendant tetrahedron is a scaled version of one of the original six. When the Morton space-filling curve \cite{Morton66, TropfHerzog81} is applied to the refinement of the initial cube, each descending cube partitions into six tetrahedral descendants at the same refinement level, or depth.
This ensures that cubic and tetrahedral refinement commute. This commutativity allows us to uniquely identify each tetrahedral descendant using a concise index tuple: its level, the coordinates of the surrounding cube descendant's lower-left-back corner (\textit{anchor coordinate}), and its type $\in \braceset{0, \ldots, 5}$. The type encodes the specific orientation of the element within its enclosing sub-cube.

For the pyramidal children, a novel refinement strategy is required.
Existing solutions for pyramidal subdivision often result in an undesirably
large number of element types or introduce entirely new element shapes
\cite{KALLINDERIS20055019, KHAWAJA20001231,
KasmaiThompsonLukeJankunMachiraju}, which significantly increases the
complexity of adaptive mesh refinement algorithms. Our design goal is a
refinement that avoids producing degenerate mesh elements, maintains a
minimal number of element types, and commutes with cubic refinement. A naive
extension of Bey's refinement to pyramids would typically involve dividing
the root cube into three similar pyramids rotated 120 degree along a long
diagonal. However, recursively applying this subdivision scheme quickly
increases the number of distinct child shapes, leading to significantly more
complex indexing and low-level algorithms.


To fulfill our requirements we propose a novel decomposition of the reference cube. We place two pyramidal elements within the cube, where one is a $180$-degree rotation of the other, leaving two tetrahedral gaps.
This arrangement guarantees that every pyramidal descendant is a scaled version of one of the original two. Furthermore, this decomposition allows us to leverage and reuse existing, established algorithms for the tetrahedral children of a pyramid, significantly reducing the overall algorithmic complexity.

Thanks to this strategy, pyramidal and cubic refinement commute, and we can define each pyramid
by its surrounding descendant cube and its type $\in \braceset{6, 7}$.
With this convention, the shape of any descendant of the root pyramid is
implicit in its type.
In analogy to tetrahedral elements, we therefore uniquely define a pyramid
by its \textit{anchor coordinate} $\vec{c}_0$ (see \figref{pyra_types}), its
level, and its type.


\begin{figure}
	\centering
	\begin{subfigure}[t]{.32\textwidth}
		\centering
		\resizebox{!}{0.55\textwidth}{
\begingroup%
  \makeatletter%
  \providecommand\color[2][]{%
    \errmessage{(Inkscape) Color is used for the text in Inkscape, but the package 'color.sty' is not loaded}%
    \renewcommand\color[2][]{}%
  }%
  \providecommand\transparent[1]{%
    \errmessage{(Inkscape) Transparency is used (non-zero) for the text in Inkscape, but the package 'transparent.sty' is not loaded}%
    \renewcommand\transparent[1]{}%
  }%
  \providecommand\rotatebox[2]{#2}%
  \newcommand*\fsize{\dimexpr\f@size pt\relax}%
  \newcommand*\lineheight[1]{\fontsize{80}{44}\selectfont}%
  \ifx\svgwidth\undefined%
    \setlength{\unitlength}{487.26187134bp}%
    \ifx\svgscale\undefined%
      \relax%
    \else%
      \setlength{\unitlength}{\unitlength * \real{\svgscale}}%
    \fi%
  \else%
    \setlength{\unitlength}{\svgwidth}%
  \fi%
  \global\let\svgwidth\undefined%
  \global\let\svgscale\undefined%
  \makeatother%
  \begin{picture}(1,0.61323149)%
    \lineheight{1}%
    \setlength\tabcolsep{0pt}%
    \put(0,0){\includegraphics[width=\unitlength,page=1]{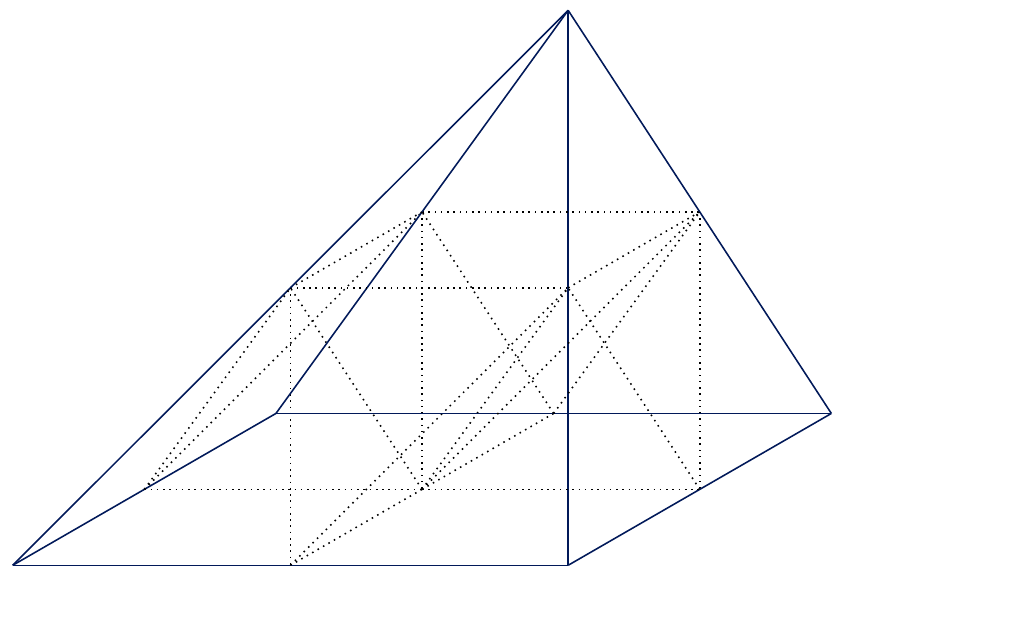}}%
    \put(0.25972155,0.16112353){\color[rgb]{0,0,0}\makebox(0,0)[lt]{\lineheight{1.25}\smash{\begin{tabular}[t]{l}$\vec{c}_0$\end{tabular}}}}%
    \put(0.01047428,0.01394795){\color[rgb]{0,0,0}\makebox(0,0)[lt]{\lineheight{1.25}\smash{\begin{tabular}[t]{l}$\vec{x}_1$\end{tabular}}}}%
    \put(0.53996435,0.01298664){\color[rgb]{0,0,0}\makebox(0,0)[lt]{\lineheight{1.25}\smash{\begin{tabular}[t]{l}$\vec{x}_3$\end{tabular}}}}%
    \put(0.80990461,0.15642242){\color[rgb]{0,0,0}\makebox(0,0)[lt]{\lineheight{1.25}\smash{\begin{tabular}[t]{l}$\vec{x}_2$\end{tabular}}}}%
    \put(0.58989212,0.57066105){\color[rgb]{0,0,0}\makebox(0,0)[lt]{\lineheight{1.25}\smash{\begin{tabular}[t]{l}$\vec{x}_4$\end{tabular}}}}%
  \end{picture}%
\endgroup%
}
		\caption{The refinement of a pyramid of type 6.}
		\figlab{pyra_refinment_type_6}
	\end{subfigure}
	\begin{subfigure}[t]{.32\textwidth}
		\centering
		\resizebox{!}{0.51\textwidth}{
\begingroup%
  \makeatletter%
  \providecommand\color[2][]{%
    \errmessage{(Inkscape) Color is used for the text in Inkscape, but the package 'color.sty' is not loaded}%
    \renewcommand\color[2][]{}%
  }%
  \providecommand\transparent[1]{%
    \errmessage{(Inkscape) Transparency is used (non-zero) for the text in Inkscape, but the package 'transparent.sty' is not loaded}%
    \renewcommand\transparent[1]{}%
  }%
  \providecommand\rotatebox[2]{#2}%
  \newcommand*\fsize{\dimexpr\f@size pt\relax}%
  \newcommand*\lineheight[1]{\fontsize{40}{34}\selectfont}%
  \ifx\svgwidth\undefined%
    \setlength{\unitlength}{532.84785461bp}%
    \ifx\svgscale\undefined%
      \relax%
    \else%
      \setlength{\unitlength}{\unitlength * \real{\svgscale}}%
    \fi%
  \else%
    \setlength{\unitlength}{\svgwidth}%
  \fi%
  \global\let\svgwidth\undefined%
  \global\let\svgscale\undefined%
  \makeatother%
  \begin{picture}(1,0.51747151)%
    \lineheight{1}%
    \setlength\tabcolsep{0pt}%
    \put(0,0){\includegraphics[width=\unitlength,page=1]{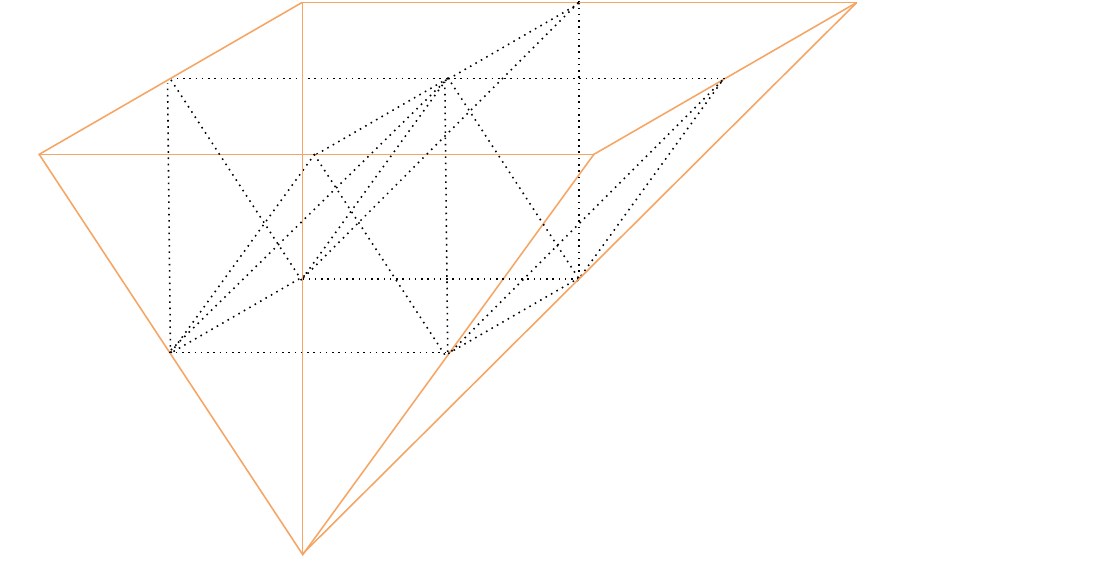}}%
    \put(0.28529624,0.00646584){\color[rgb]{0,0,0}\makebox(0,0)[lt]{\lineheight{1.25}\smash{\begin{tabular}[t]{l}$\vec{c}_0$\end{tabular}}}}%
    \put(-0.00255668,0.33668661){\color[rgb]{0,0,0}\makebox(0,0)[lt]{\lineheight{1.25}\smash{\begin{tabular}[t]{l}$\vec{x}_5$\end{tabular}}}}%
    \put(0.52971884,0.32804798){\color[rgb]{0,0,0}\makebox(0,0)[lt]{\lineheight{1.25}\smash{\begin{tabular}[t]{l}$\vec{x}_7$\end{tabular}}}}%
    \put(0.16600356,0.482287){\color[rgb]{0,0,0}\makebox(0,0)[lt]{\lineheight{1.25}\smash{\begin{tabular}[t]{l}$\vec{x}_4$\end{tabular}}}}%
    \put(0.7808045,0.47668309){\color[rgb]{0,0,0}\makebox(0,0)[lt]{\lineheight{1.25}\smash{\begin{tabular}[t]{l}$\vec{x}_6$\end{tabular}}}}%
  \end{picture}%
\endgroup%
}
		\caption{The refinement of a pyramid of type 7.}
		\figlab{pyra_refinment_type_7}
	\end{subfigure}
	\begin{subfigure}[t]{.32\textwidth}
		\centering
		\resizebox{!}{0.8\textwidth}{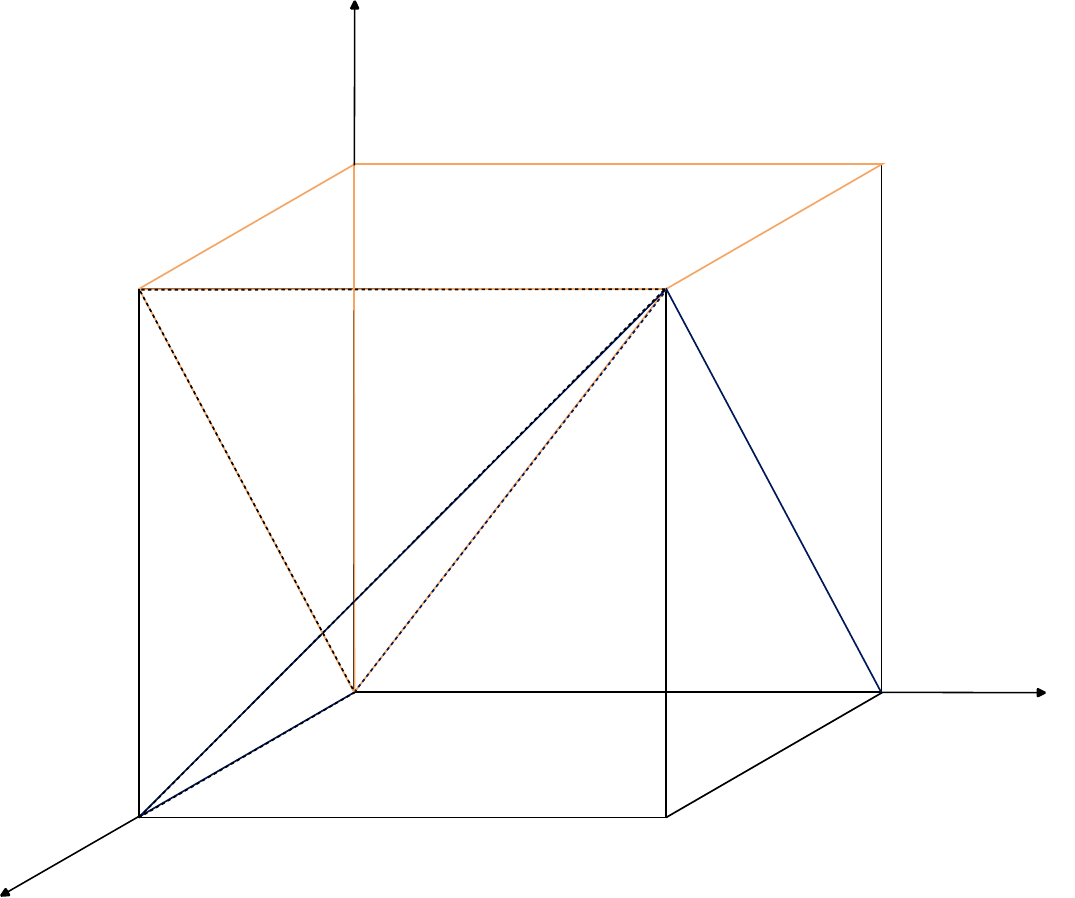}
		\caption{Each child is of  one of four types.}
		\figlab{pyra_types}
	\end{subfigure}
        \caption{%
(a): The vertices of a pyramid of type $6$ and its ten children.
The vertex $\vec{c}_0$ is the anchor coordinate of the parent.
The pyramid refines into six pyramids, one of them upside down, and four
tetrahedra of two types.
(b): The vertices of a pyramid of type $7$ and its ten children.
(c): the two different types of the pyramids and the tetrahedra,
respectively, that occur as children, partition the cube. Defining the refinement this way, the cube (c) is also a sub-cube of (a) and (b).
Note that both types of pyramids share the same anchor coordinate. Only the pyramid in the middle switches the type, that is the center pyramid of (a) is a pyramid of type 7 and the center pyramid of (b) is a pyramid of type 6.
}%
\figlab{pyra_refinement_types}%
\end{figure}%

Upon refinement, we divide a pyramid into ten children: six pyramids and four tetrahedra, as displayed in \figref{pyra_refinement_types}.
Each child can be identified by a shift of its parent's anchor coordinate (see Table \ref{tab:coord-shift}) and its type, which also depends on that of the parent (see Table \ref{tab:child-type}).

\begin{table}
	\begin{tabular}{ccccccccccc}
  \toprule
      	 Child & $0$ & $1$ & $2$ & $3$ & $4$ & $5$ & $6$ & $7$ & $8$ & $9$ \\
  \midrule
         Parent of type $6$ & $0$ & $\vec{h}_x$ &
$\vec{h}_x$ & $\vec{h}_y$ & $\vec{h}_y$ &
$\vec{h}_{xy}$ & $\vec{h}_{xy}$ &
$\vec{h}_{xy}$ & $\vec{h}_{xy}$ &
$\vec{h}_{xyz}$ \\
         Parent of type $7$ & $0$ & $\vec{h}_z$ &
$\vec{h}_z$ & $\vec{h}_z$ & $\vec{h}_z$ &
$\vec{h}_{xz}$ & $\vec{h}_{xz}$ &
$\vec{h}_{yz}$ & $\vec{h}_{yz}$ &
$\vec{h}_{xyz}$ \\
  \bottomrule
	\end{tabular}
        \caption{This table defines how to shift the anchor coordinate of a
pyramid of type $6$ or $7$ to compute the anchor coordinate of one of its
children. $\vec{h}_{\{\text{axes}\}}$ refers to a shift along all the given
axes for a pyramid of length $h$, yet to be divided by two.}%
\tablab{coord-shift}%
\end{table}

\begin{table}
  \begin{tabular}{ccccccccccc}
    \toprule
      Child & $0$ & $1$ & $2$ & $3$ & $4$ & $5$ & $6$ & $7$ & $8$ & $9$ \\
    \midrule
    Parent of type 6 &
              $6$ & $3$ & $6$ & $0$ & $6$ & $0$ & $3$ & $6$ & $7$ & $6$ \\
    Parent of type 7 &
              $7$ & $0$ & $3$ & $6$ & $7$ & $3$ & $7$ & $0$ & $7$ & $7$ \\
    \bottomrule
  \end{tabular}
  \caption{The type for every child of a pyramid of type $6$ and 7. In each
sub-cube we order the elements by its type. Therefore, the tetrahedral
children (types $0$ to $5$) in a sub-cube are the predecessors of the
pyramidal children (types $6$ and $7$).}
\tablab{child-type}%
\end{table}%

The original pyramid, which has no parent, has level $0$ and is called the
\textit{root pyramid} $P^0$.
%
%
For practical purposes, we map the root pyramid into a reference cube of
length $2^\mathcal{L}$, see pyramid $P_6$ in \figref{pyra_types}.
Therefore, every element's anchor has integer coordinates up to the maximal
level $\mathcal L$. Altogether we define the set of elements in a pyramidal
refinement tree as follows:

\begin{defi}
    $\mathcal{P} := \{ P ~|~ P \text{ is a descendant of } P^0\text{, with } 0 \leq l(P) \leq \mathcal{L} \}$.
    \label{pyramid_set}
\end{defi}
Similar to other elements, the pyramidal space-filling curve uses
the interleaving of the coordinates and the type of the pyramid. Whenever we are talking about the coordinates we refer to their binary representation with $\mathcal L$ bits. Additionally, we use $X$, $Y$, and $Z$ to denote the respective parts of the anchor coordinate of an element.

As outlined above, we inherit six types ($0$ to $5$) for the tetrahedral children
and add the types $6$ and $7$ for the pyramidal children (see Figure
\ref{fig:pyra_types}). We define $P^i$ as the $i$-th descendant along a branch towards an element $P$. Furthermore we split the type of an element in the pyramidal refinement tree into three tuples.

\begin{defi}
	\label{Def:types}
        Let $P \in \mathcal{P}$ be an element of level $l$.
$B(P)$ is defined as the $\mathcal{L}$-tuple consisting of the types of
$P$'s ancestors in the first $l$ entries, starting with the root pyramid
$P^0$.
The last $\mathcal{L}-l$ entries of $B(P)$ are zero:
  \begin{align}
	B = B(P) &= \left(\ptype(P^1), \ptype(P^2), \ldots, \ptype(P),
                    0, \ldots, 0 \right)\\
	         &= (b_{\mathcal{L}-1}, \ldots, b_0)
                    \in \left\{0, \ldots, 7\right\}^\mathcal{L} \nonumber,
  \end{align}
	where the entry $b_i$ is given by:
	\begin{equation}
	b_i= \begin{cases}
	\ptype(P^{\mathcal{L}-i}), &\mathcal{L}-1\geq i \geq \mathcal{L}-l,\\
	0, &\mathcal{L}-l > i \geq 0.
	\end{cases}
	\end{equation}
\end{defi}
In the refinement tree of a pyramid we find eight different types of
elements, that is, two pyramidal and six tetrahedral shapes.
Therefore, we can rewrite each entry $b_j$ as a 3-digit binary number
$b_j=(b_j^2, b_j^1, b_j^0)_2$.
This gives rise to three new $\mathcal{L}$-tuples $B^2, B^1$ and $B^0$ with
\begin{equation}
\begin{aligned}
B^0 &= \left(b^0_{\mathcal{L}-1}, \ldots, b^0_0\right),\\
B^1 &= \left(b^1_{\mathcal{L}-1}, \ldots, b^1_0\right),\\
B^2 &= \left(b^2_{\mathcal{L}-1}, \ldots, b^2_0\right).
\end{aligned}
\end{equation}

Eventually, we define the \textit{pyramid index} as the bit-wise interleaving of these
tuples, concatenating the first bit of each in turn, then the second of
each, etc.:
\begin{defi}
        The \emph{pyramid index} of an element $P\in \mathcal{P}$ is given as
the interleaving of the $\mathcal{L}$-tuples $Z$, $Y$, $X$, $B^2$, $B^1$,
and $B^0$, written
	\begin{align}
	m_P(P) := Z \interleaving Y \interleaving X \interleaving
                  B^2 \interleaving B^1 \interleaving B^0 .
	\end{align}
\end{defi}
Consequently, for the same anchor node, the elements are ordered by their type. Hence, a type-$0$ tetrahedra comes before a tetrahedra
of type $3$, followed by the pyramids of type $6$ and type $7$; cf.\
\figref{pyra_types}.

In summary, we can uniquely identify an element in the pyramidal refinement
using its pyramid index and its level.
In effect, we construct the pyramid index such that we are able to prove
that it is an SFC index according to the following definition from \cite{Holke18}:
\begin{defi}
	\label{Def:SFC-index}
	Suppose we have a set of $\mathcal{S}$ elements in the refinement of a pyramid.
        An \emph{SFC index} $\mathcal{I}$ is a map
$\mathcal{I}:\mathcal{S}\rightarrow \mathbb{N}_0$ with the following
properties:
	\begin{enumerate}
                \item Restricted to a level $l$, $\mathcal{I}$ is injective:
$\mathcal{I} \times l: \mathcal{S} \rightarrow \mathbb{N}_0 \times
\mathbb{N}_0$ is injective.
                \item If $E$ is an ancestor of $E'$ then $\mathcal{I}(E)
\leq \mathcal{I}(E')$ (refining is monotonous).
                \item $\mathcal{I}(E) < \mathcal{I}(\hat{E})$ and $\hat{E}$
not a descendant of $E$, then $\mathcal{I}(E) \leq \mathcal{I}(E') <
\mathcal{I}(\hat{E})$ for all descendants $E'$ of $E$ (refining is local).
	\end{enumerate}
\end{defi}

We show these three properties by embedding the pyramid index into a
six-dimensional Morton index.

\begin{theorem}
	\label{Thm:pyra-sfc}
	The pyramid index $m_P$ is an SFC index.
\end{theorem}
We reduce the proof to the fact that the Morton index in any dimension is
formally an SFC index \cite{BursteddeHolke16}.
Thus, we identify the anchor coordinates with the first three coordinates of
the anchor node of the cube,
and we map the three type-representing tuples $B^0$, $B^1$, and $B^2$ onto
the last three coordinates of the anchor node.
Since $0 \leq B^i < 2^\mathcal{L}$, we can use the set of children
$Q_0 := [0,2^\mathcal{L}]^6$ to embed the pyramids.


Given the anchor node of a cube and its level, the set $\mathcal{Q}$ of all
six-dimensional cubes in the refinement is given as $\mathcal{Q}=Q_{(x_0, \ldots,
x_5), l}$, with $x_0, \ldots, x_5$ as the coordinates of the anchor node.
The Morton index for six-dimensional cubes is given by the bit-wise interleaving
of the coordinates of the cube as a generalization of the well-known 2D and 3D
indices \cite{Morton66, TropfHerzog81}:
\begin{equation}
  m_Q(Q) = X^5 \interleaving X^4 \interleaving X^3 \interleaving X^2\interleaving X^1\interleaving X^0 .
\end{equation}
We use Proposition \ref{Pro:sic-dim-mort} to map the pyramid index to six-dimensional cubes.
\begin{proposition}
	\label{Pro:sic-dim-mort}
	The map
        \begin{equation}
	\begin{aligned}
	\Theta : \mathcal{P} &\rightarrow \mathcal{Q} , \\
	P &\mapsto Q_{(B^2(P), B^1(P), B^0(P), x(P), y(P), z(P)), l(P)}
	\end{aligned}
        \end{equation}
	is injective and satisfies
	\begin{equation}
	m_Q(\Theta(P)) = m_P(P) .
	\end{equation}
        Furthermore, it ensures the property that $P'$ is a child of $P$ if
and only if $\Theta(P')$ is a child of $\Theta(P)$.
\end{proposition}
\begin{proof}[Proof of Proposition \ref{Pro:sic-dim-mort}]
	The equation $m_P(P)=m_Q(\Theta(P))$ follows from the definition of the indices on $\mathcal{P}$ and $\mathcal{Q}$. Using that an element is uniquely determined by its index and its level, we see that $\Theta$ is injective.
	To show the second property, let $P, P'\in \mathcal{P}$, where $P'$ is a child of $P$. Additionally let $l = l(P)$. We use that $Q' := \Theta(P')$ is a child of $Q:=\Theta(P)$ if and only if for each $i\in\left\{0, \ldots, 5\right\}$ it holds that
	\begin{align}
	\label{Equ:Incube}
	x_i(Q')\in\left\{x_i(Q), x_i(Q)+2^{\mathcal{L}-(l+1)}\right\}.
	\end{align}
        We know that \eqnref{Incube} holds for the $x$-, $y$-, and
$z$-coordinate of the anchor node of $P'$, meaning it holds for $i\in
\left\{3,4,5\right\}$.
By definition $B^j(P')$ is the same as $B^j(P)$ except for the new type at
position $\mathcal{L}-(l+1)$, for which
\begin{equation}
\begin{aligned}
  B^j(P')_{\mathcal{L}-(l+1)} & =
    b^j_{\mathcal{L}-(l+1)}(P')\in \left\{0,1\right\} \quad\text{and} \\
   B^j(P)_{\mathcal{L}-(l+1)} & = 0 .
\end{aligned}
\end{equation}
In conclusion, \eqnref{Incube} holds for $i\in\left\{0,1,2\right\}$, and
the first implication is shown.

        For the reverse implication we assume that $\Theta(P')$ is a child
of $\Theta(P)$. Since the level of $P'$ is bigger than $0$, we know that $l(\Theta(P'))>0$
and $P'$ has a parent $S$. By the uniqueness of the parent of a cube the
identity $\Theta(S) = \Theta(P)$ must hold. Thus, we obtain $S=P$, since
$\Theta$ is injective and $P'$ is the child of $P$.
\end{proof}

On this basis, we proceed to prove Theorem \ref{Thm:pyra-sfc}.
\begin{proof}[Proof of Theorem \ref{Thm:pyra-sfc}]
        The first two properties of an SFC index can be deduced by the fact
that the Morton index for six-dimensional cubes is an SFC index together
with Proposition \ref{Pro:sic-dim-mort}.
        Using the last property of Proposition \ref{Pro:sic-dim-mort} inductively, we use that $P'$ is a
descendant of $P$ if and only if $\Theta(P')$ is a descendant of
$\Theta(P)$ to guarantee the third property of an SFC index.
\end{proof}

\begin{SCfigure}
\centering
	\includegraphics[width=0.3\textwidth]{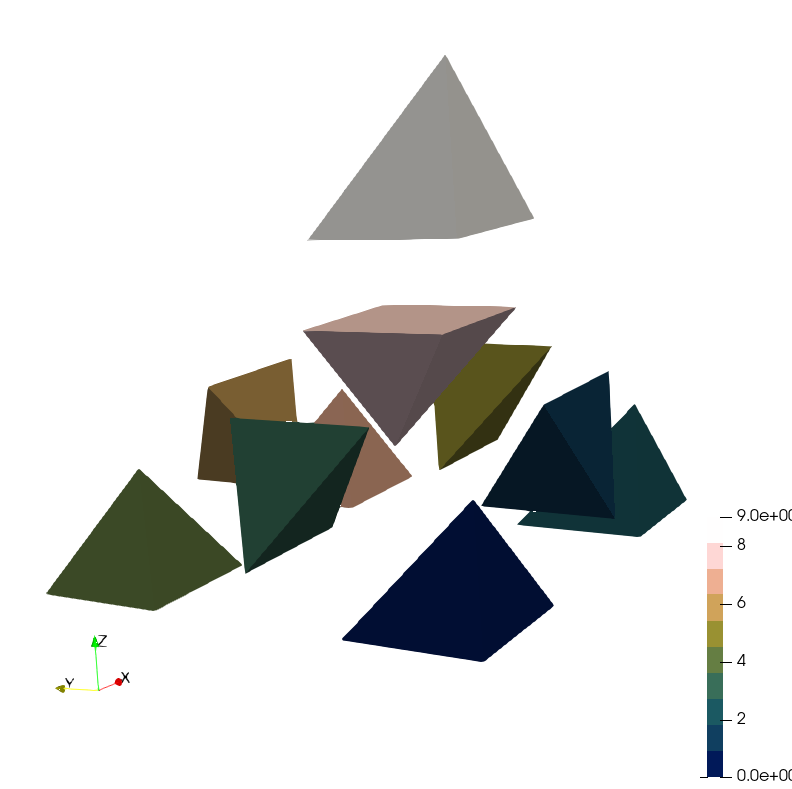}
        \caption{A single pyramid that is refined once. Each child is colored
according to its index. The element with the index $0$ is colored red, the last
element is colored blue.}
\figlab{pyra-sfc}%
\end{SCfigure}%

With the help of Figure \ref{fig:pyra-sfc} we define a \textit{local index} of an
element. The local index $i$ refers to the $i$-th child of an element according
to the SFC index. For a parent of type $6$ the children are already ordered by
the SFC, while the local index for parents of type $7$ is not straightforward;
see \tabref{local-index}.
\begin{table}
	\begin{tabular}{ccccccccccc}
		\toprule
                Child & $P_0$ & $T_1$ & $P_2$ & $T_3$ & $P_4$ &
                        $T_5$ & $T_6$ & $P_7$ & $P_8$ & $P_9$ \\
        \midrule
        Parent of type 6 &
                        $0$ & $1$ & $2$ & $3$ & $4$ & $5$ & $6$ & $7$ & $8$ & $9$ \\
        Parent of type 7 &
                        $8$ & $6$ & $3$ & $9$ & $7$ & $1$ & $2$ & $5$ & $4$ & $0$ \\
		\bottomrule
	\end{tabular}
        \caption{The local index for every child of a pyramid of type $6$ is
                 consecutive by construction, while the order is non-trivial
                 for a parent of type 7.
}%
\tablab{local-index}%
\end{table}%
Lastly, we observe the following two Lemmata by analyzing the
refinement of a pyramid (see Figure \ref{fig:pyra-sfc}) and the refinement
of type $0$ or $3$ tetrahedra, which are the only tetrahedral children of a
pyramid.
\begin{lemma}
        The five pyramidal children touching a face of their parent lie in a
corner of their parent.\label{lem:pyra_corner}
\end{lemma}
\begin{lemma}
        A descendant of a tetrahedron in a pyramidal refinement tree that
has a pyramid as a neighbor lies in a corner of its parent.
\label{lem:tet_corner}
\end{lemma}

A visualization of both lemmata can be seen in \figref{pyra_tet}.
\begin{SCfigure}
	\resizebox{50mm}{!}{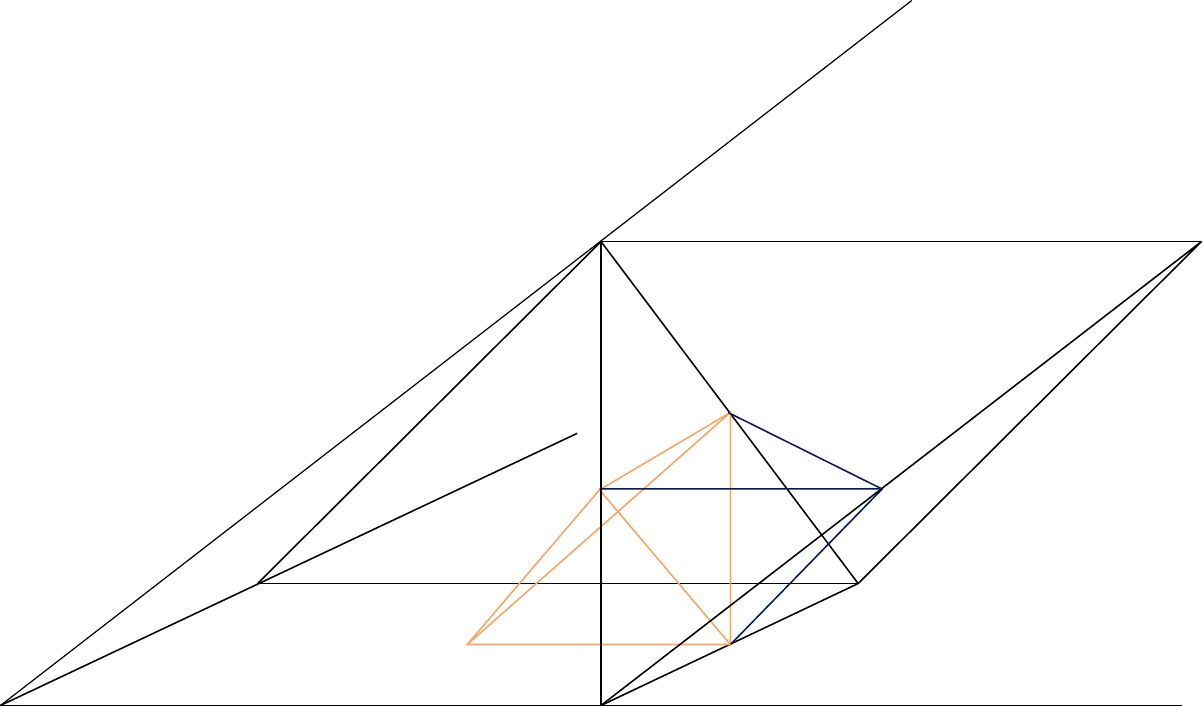}
	\caption{A pyramid and its neighbor along $f_3$. In the refinement of the black (outer) pyramid, all children that do not have type $3$ or $0$ only have tetrahedral neighbors within the same tree.\label{fig:pyra_tet}}
\end{SCfigure}

\section{Algorithms at the element level}
\seclab{lowlevel}

In this section we give an introduction into the element-wise (low-level)
algorithms for pyramids with a focus on the differences to the corresponding algorithms for other element shapes.
Major differences for pyramids take effect when constructing the parent, as
discussed in \secref{parent}.
\secref{child} is a brief example of a typical low-level algorithm, namely constructing the child element of a given index.
We proceed constructing neighboring elements inside the same pyramidal tree
in \secref{within-root-neighbor}.
Neighbor elements across trees implicate a transformation and possibly a
translation into a different shape, which we cover in
\secref{outside-root-neighbor}.

\subsection{Parent}
\label{Sec:parent}

For an element $E$, we can compute the ancestor of level $l(E)-1$ using the \texttt{parent} algorithm. Given the type of $E$ and its anchor coordinate, we can compute the type of the parent. The level and the coordinates of the parent are given implicitly, by setting the $l(E)$-th bit of $E$'s coordinate to zero.
We refer to this operation as ``cutting'', especially when multiple bits of the coordinate are set to zero.

For a tetrahedron $T$ descending from a pyramidal root element, the computation of the type is not trivial. In contrast to a pure tetrahedral tree (or other trees containing only one element shape), we cannot determine the type of the parent without further analysis.

One remedy would be to recompute the $\mathcal{L}$-tuple $B$ (see Def.\
\ref{Def:types}), descending from the root pyramid down to level $l(E)-1$, to
check at which level the last pyramidal ancestor of the tetrahedron occurs.
This, however, is a rather costly computation; to avoid it, we introduce the
field $\switchShapeLevel$: For each element, it stores the smallest level at
which the ancestor has the shape of a tetrahedron. We can compute this field
during the construction of an element without using the iterative
computation of the types of all ancestors.
For all pyramids we set $\switchShapeLevel=-1$, since they do not have any
tetrahedral ancestor at all.

For a tetrahedron with a pyramidal parent,
the type of the parent can be
determined by comparing the $z$-coordinates:
\begin{itemize}
    \item If the tetrahedron has the same $z$-coordinate as the anchor coordinate of the parent, the parent has type $6$.
    \item
    Otherwise, the parent has type $7$.
\end{itemize}
Aside from this computation of the type of the parent, the \texttt{parent} algorithm does not change for elements in a pyramidal refinement, see Algorithm \ref{Alg:Parent}.

\begin{algorithm}
	\TitleOfAlgo{\texttt{t8\_dpyra\_parent}}
    \setlength{\algoCommentLength}{5.0cm}
    \DontPrintSemicolon
	\KwData{An element $E$ and an empty element $P$, which is going to be filled with the data of the parent of $E$.}
    \myAlgoLine{Set level of parent $P$}
    {$P.\level = E.\level - 1$}

	\uIf{\textnormal{$E$ is a pyramid}}{

        \myAlgoBlockHeader{$P$ is a pyramid}

        \myAlgoLine{Compute the cube id of E}
            {$c_{id} \mygets{36pt} \texttt{cube\_id}(E, E.\level)$}

        \myAlgoLine{Determine type of $P$}
            {$P.\type \mygets{36pt} \texttt{from}(E.\type, c_{id})$}

        \myAlgoLine{Set $P$'s coordinates via cutting}
            {$P.\coord \mygets{36pt} \texttt{cut}(E.\coord)$}


        \myAlgoLine{Set $\switchShapeLevel$ to invalid}
            {$P.\switchShapeLevel = -1$}
		
		}
    \Else{
	\uIf{$E.\switchShapeLevel = E.\level$}
		{
			\myAlgoBlockHeader{$P$ is a pyramid}

            \myAlgoLine{Determine type of $P$}
            {$P.\type \mygets{36pt} \texttt{from}(E.\coord)$}

            \myAlgoLine{Set $P$'s coordinates via cutting}
                {$P.\coord \mygets{36pt} \texttt{cut}(E.\coord)$}
			
			
			\myAlgoLine{Set $\switchShapeLevel$ to invalid}
                {$P.\switchShapeLevel = -1$}
		}
	\Else{
        \myAlgoBlockHeader{$P$ is a tetrahedron}

        \myAlgoLine{Compute $P$ as tetrahedral parent}
            {$P = \texttt{tet\_parent}(E)$}
		
        \myAlgoLine{Set lowest-level tetrahedral ancestor}
            {$P.\switchShapeLevel = E.\switchShapeLevel$\\}
	}
    }
	\caption{The algorithm to compute the parent of an element. To avoid the recomputation of the types of all ancestors, we introduce a new field $\switchShapeLevel$ that can be used to look up the smallest level at which an element's ancestor is a tetrahedron.
    \label{Alg:Parent}}
\end{algorithm}

%
%
%
%
		
%
%
%
%
%

\subsection{Child}
\label{Sec:child}

The child algorithm computes the $i$-th child of an input element $E$, with $0 \leq i < N_c$, where $N_c$ denotes the maximal number of children of $E$.
%
We can subdivide the problem into two cases:

\enumerateCases{
\item
If $E$ is a pyramid, we shift the anchor coordinate based on the type of the input element by $0$ or $2^{\mathcal{L}-(l+1)}$, with $l=l(E)$. Using the type of $E$ and the local id $i$, we can determine the type of the child. Altogether,  we can determine the child using the computed coordinates, the type, and the implicitly given level.
\item
If $E$ is a tetrahedron descending from a pyramid, we can reuse the algorithm implemented for tetrahedra, since the types of pyramids and tetrahedra do not interfere. Furthermore, the computation of the anchor coordinate is the same for pyramids and tetrahedra.
}
We use Tables \ref{tab:coord-shift} and \ref{tab:child-type} to look up these values for pyramids, see Algorithm \ref{Alg:Child}.

\begin{algorithm}
	\TitleOfAlgo{\texttt{t8\_dpyra\_child}}
    \setlength{\algoCommentLength}{5.5cm}
   \DontPrintSemicolon
    \small
	\KwData{An element $E$, a local index $i$, and an empty element $C$, which is going to be filled with the data of the $i$-th child of $E$.}

    \myAlgoLine{Set level of the child $C$}
    {$C.\level = E.\level + 1$}
	\uIf{\textnormal{$E$ is a pyramid}}{

        \myAlgoLine{Shift the coordinates using Table \ref{tab:coord-shift}}
		{$C$.\coord \mygets{36pt} $E$.coord + \texttt{shift}(${i}$)}

        \myAlgoLine{Set the type using Table \ref{tab:child-type} }
        {$C$.type = \texttt{child\_type}($E$.type, $i$)}

        \If{\textnormal{$C$ is a tetrahedron}}{
            \myAlgoLine{Set $\switchShapeLevel$}
            {$C.\switchShapeLevel = C.\level$ }
        }
	}
	\Else{
        \myAlgoBlockHeader{$E$ is a tetrahedron}

        \myAlgoLine{Call the tetrahedral routine}
        {$C$ = \texttt{t8\_tet\_child($E$, $i$)}}

        \myAlgoLine{$\switchShapeLevel$ is the same for the child}
		{$C$.$\switchShapeLevel$ = $E$.level}
    }
	\caption{The algorithm to compute the $i$-th child of an element. We use the Tables \ref{tab:coord-shift} and \ref{tab:child-type} for the computation of the anchor coordinate and the type of the child.
    \label{Alg:Child}}
\end{algorithm}

\subsection{Same-tree neighbor elements}
\seclab{within-root-neighbor}
This section  describes how to compute the neighbors of an element inside of a tree. In \secref{outside-root-neighbor} we describe the computation of neighbors over tree boundaries.
The computation of ghost elements for one process requires the computation of
an element's neighbor along a given face. We divide the problem into
three sub-tasks and solve them subsequently:
\enumerateSteps{
	\item Define the shape of the neighbor. The neighbor of a pyramid may be another pyramid
	or a tetrahedron; analogously, a tetrahedron's neighbor may either be another tetrahedron or a pyramid.
	\item Compute the coordinates and the type of the neighbor.
	\item Define the number of the dual face, i.e., the face of
	the neighbor touching the original element.
}

In order to identify the element's faces, we define a numbering of the corners
and the faces of a pyramid. To be consistent with previous work, we choose
the labels given in \cite{BursteddeHolke17}; see Figure \ref{fig:pyra_face}. We use $f_i$ to denote the $i$-th face of an element, or $f$ for a face in general. We use $C$ to address a corner and $C_i$ for the $i$-th corner.
Note that we only define the numbering on a type-$6$ pyramid, because a pyramid
of type $7$ can be seen as a $180$-degree rotation of a type-$6$ pyramid.

\begin{figure}
    \begin{floatrow}
        \ffigbox{
        \def\svgwidth{50mm}{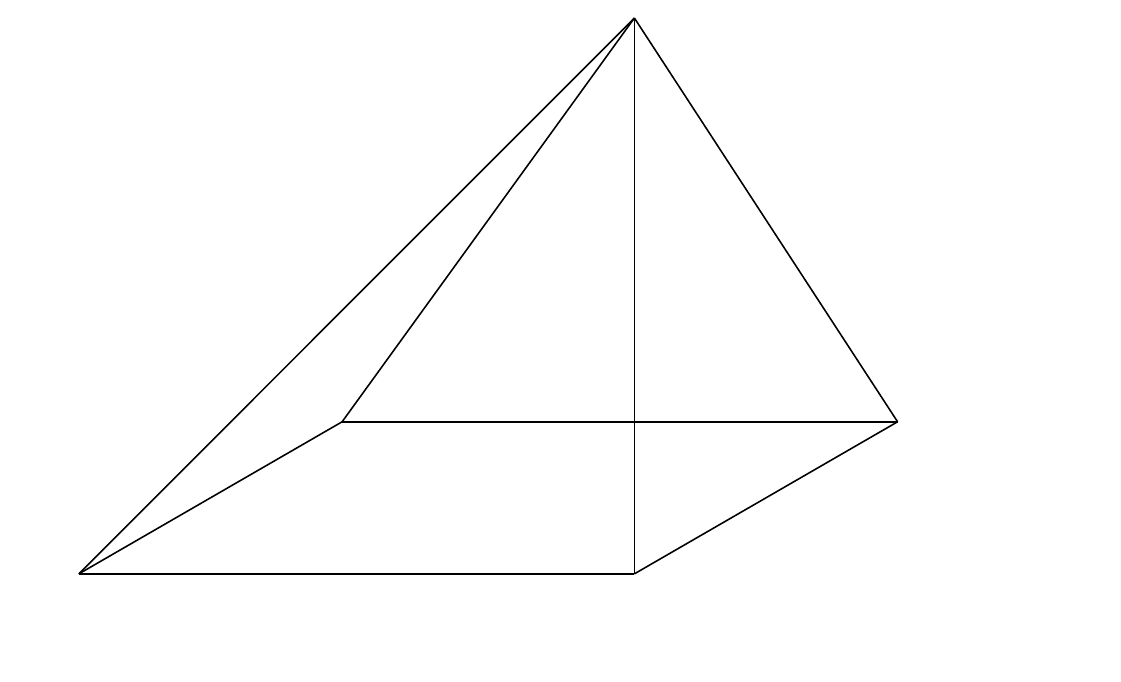}%
        }
        {\caption{The corners of a pyramid.}
        \label{fig:pyra_face}}

        \capbtabbox{
		\begin{tabular}{ccccc}
				\toprule
				Face & \multicolumn{4}{c}{corners} \\
				\midrule
				$f_0$ & $C_0$ & $C_3$& $C_4$ & $-$ \\
				$f_1$ & $C_1$ & $C_2$& $C_4$ & $-$ \\
				$f_2$ & $C_0$ & $C_1$& $C_4$ & $-$ \\
				$f_3$ & $C_2$ & $C_3$& $C_4$ & $-$ \\
				$f_4$ & $C_0$ & $C_1$& $C_2$ & $C_3$ \\
				\bottomrule
			\end{tabular}
        }
        {\caption{The faces of a pyramid.}%
        \label{Tab:pyra_face}}%
    \end{floatrow}
\end{figure}%

\ourSubSubsection{Shape of the neighbor}
In the following, we show how to determine the shape of the same-tree neighbor.
For a pyramid, the quadrilateral face $f_4$ is connected to another pyramid,
while all other neighbors are always tetrahedra.
Note that the tetrahedral children of a refined pyramid always have type $0$ or
$3$, so only these can have a pyramid as a same-tree (and same-level) neighbor.
Therefore, we can use the already existing algorithms for all tetrahedra that do not
have type $0$ or $3$; see \cite{BursteddeHolke16}.

For a tetrahedron of type $0$ or $3$, we distinguish two cases when identifying the neighbor's shape:
\enumerateCases{
	\item The parent of the tetrahedron is a pyramid. In that case, the tetrahedron has three faces touching its pyramidal siblings and one face connected to another tetrahedron. This face's number depends on the type of the pyramidal parent and the local id of the element itself; it is defined in \tabref{tet-face}.
	\item The parent of the tetrahedron is another tetrahedron.
    It may have many tetrahedral ancestors, the most interesting being the tetrahedral ancestor with the lowest level,
    because its parent is a pyramid. If the input element touches a face of this lowest-level tetrahedral ancestor, we can reduce this case to the first case.
}
\begin{table}
\begin{subtable}{.48\textwidth}
\centering
    \begin{tabular}{ccc}%
                        \toprule
						Local id & pyramid $6$ & pyramid $7$ \\
						\midrule
						$1$ &   $1$ &   $3$ \\
						$2$ & $-$ &   $3$ \\
						$3$ &   $1$ & $-$ \\
						$5$ &   $0$ &   $2$ \\
						$6$ &   $0$ & $-$ \\
						$7$ & $-$ &   $2$ \\
      \bottomrule
    \end{tabular}%
    \caption{The faces of a tetrahedron inside a pyramid which extrude to a tetrahedron.}%
    \tablab{tet-face}%
\end{subtable}
\begin{subtable}{.48\textwidth}
\centering
    \begin{tabular}{cccc}%
						\toprule
                                                Face
						& $CTF_0$ & $CTF_1$& $CTF_2$ \\
						\midrule
						$f_0$ & $1$ & $2$ & $3$ \\
						$f_1$ & $0$ & $2$ & $3$ \\
						$f_2$ & $0$ & $1$ & $3$ \\
						$f_3$ & $0$ & $1$ & $2$ \\
      \bottomrule
    \end{tabular}
    \caption{The types of the three children touching a given
              face of a tetrahedron.}%
     \tablab{child-face}%
\end{subtable}
\caption{The face connectivity between pyramids and tetrahedra.}%
\end{table}
In conclusion, for both cases we need to know which face of a tetrahedron is
touching a pyramid and which one is not. The distinction is simplified by
introducing the following definition of \textit{valid faces}:

\begin{defi}
    Let $f$ be a face of a tetrahedron with pyramidal parent;
    if $f$ is connected to a pyramid, it
    is called a valid face. \label{Def:valid_face}
\end{defi}
It is important to realize that the problem of determining the shape of a tetrahedron's neighbor is not solved yet. In the above case (b), knowing an element is connected to a valid face is not
sufficient to determine the shape of the neighbor.
To derive the necessary additional checks, we observe that for
a type-$3$ tetrahedron refined multiple times, the type of the
children not in their parent's corner alternates between $2$ and
$3$; since a type-$2$ tetrahedron can never touch a pyramid, this implies that only the children in the corner can have a pyramid as a neighbor.

Therefore, we check whether all ancestors up
to the lowest-level tetrahedral ancestor are a child in the corner of
their parent. If so, it follows that these ancestors and the element have the same type and
the same face numbering, because the children at the corner always have the
same type as their parents; see \cite{BursteddeHolke16}. Hence, we have to
check whether the given face of the ancestor is valid to determine the shape
of the neighbor of the input tetrahedron.  Now these arguments combine into the
following theorem.
\begin{theorem}
	Let $T$ be a tetrahedron of type $0$ or type $3$ that has a pyramid as an ancestor, $A = \arg\min_l \{S.l: \mathrm{S\ tetrahedron\ and\ ancestor\ of\ T}\}$ its lowest-level tetrahedral ancestor, and $T' \in \{\mathrm{S\ ancestor\ of\ T}: S.l >= A.l\}$ any tetrahedral ancestor of $T$; then $T$ has a pyramid neighbor along face $f$ if and only if:
    \begin{enumerate}[(i)]
        \item $T$ and all $T'$  are in a corner of their parent touching face $f$ of the parent, and
        \item  face $f$ of $A$ is a valid face of $A$.
    \end{enumerate}
    \label{Thm:tet_touches_pyra}
\end{theorem}
\begin{proof}[Proof of theorem \ref{Thm:tet_touches_pyra}]
	First, assume that there is a $T'$ not touching the face $f$ of its parent. Then no child of $T'$ can have contact with this face again, so
    $T$ does not touch face $f$ of $A$.
    As a result, $T$ cannot have a neighbor outside of $A$, meaning its neighbor has to be a tetrahedron, which is a contradiction.
	Secondly, assume that there is a $T$ touching face $f$ but not lying in the corner of its parent. $T$ has an ancestor in the shape of a tetrahedron and since all children are tetrahedra, the neighbor of $T$ is a tetrahedron which is a contradiction to Lemma \ref{lem:tet_corner}.
	Eventually assume that the face $f$ of A is not a valid face. But then, the neighbor along this face can not be a pyramid by the definition of a valid face, see Definition \ref{Def:valid_face}.

        For the other direction we assume that $T$ does not have a pyramid
neighbor along face $f$. Therefore the neighbor is a tetrahedron. If its
parent is a pyramid, the tetrahedron cannot lie in the corner of the pyramid
(by Lemma \ref{lem:pyra_corner}). Using this argument recursively we see
that there is either a $T'$ not lying in the corner of its parent, or there
is a $T'$ in the corner opposite of $f$, hence not touching it. Furthermore,
$f$ has to be a valid face by the definition of it.
\end{proof}
We can identify the children of a tetrahedron touching a given face by their
local indices, see \tabref{local-index}.
Since the face number of a tetrahedron is given by the number of the opposite corner,
we can use \tabref{child-face} to check whether a child is touching a given face.
For every tetrahedral ancestor, we check the local index: if it always refers to the same corner, the neighbor might be a pyramid.
The complexity of this check grows linearly with the level. To avoid it as often as possible, we first check if the given face is a valid face, which can be done in constant time via the position of the tetrahedron in its pyramidal parent, using \tabref{tet-face}.
At this point we have solved the first problem and provide Algorithm
\ref{Alg:tetPyra} to determine for an element of type $0$ or $3$ whether the
neighbor along the given face is a pyramid.

\begin{algorithm}
	\TitleOfAlgo{\texttt{t8\_dtet\_pyra\_touch}}
	\KwData{An element $T$ of type $0$ or type $3$ and a face number $f$}
    \setlength{\algoCommentLength}{5.5cm}
   \DontPrintSemicolon
    \small
    \If{\textnormal{$T$.\level = $T$.$\switchShapeLevel$}}
    {
        \myAlgoLine{Use type check of \tabref{tet_to_boundary}}
        {\Return \texttt{tet\_pyra\_connect($T, f$)}}
    }

    \myAlgoLine{Get lowest-level tetrahedral ancestor}
    {$A$ = \texttt{tet\_anc}$(T, T.\switchShapeLevel)$}

    \myAlgoLine{Check if $A$ touches $f$}
    {valid = \texttt{tet\_pyra\_touches\_face($A$, $f$)} }
	\If{\textnormal{valid}}
	{
        \myAlgoBlockHeader{If all ancestors of $T$ until $A$ touch $f$, then $T$ touches $f$}

        \myAlgoLine{Start tests with $T$}
        {type = $T$.\type}
		\For{\textnormal{$l$ = $T$.\level; $l$ > $A$.\level; $l${-}{-}}}
		{
            \myAlgoLine{Compute the cube id on level $l$}
			     { $c_{id}$ = \texttt{cube\_id}($T$, $l$)}
            \myAlgoLine{Get the local id}
			{id = \texttt{cube\_id\_to\_Beyes}(type, $c_{id}$)}
			\If{\textnormal{\textbf{not} \texttt{child\_touches\_face}($f$, id)}}
            {
            \myAlgoLine{All children have to touch face $f$}
			{\Return \texttt{false}}
            }
            \myAlgoLine{Type might vary each level}
            {type = \texttt{parent\_type}($c_{id}$, type)}
		}
    }
		\Return valid
    \vspace{0.5em}
	\caption{The algorithm to decide for a tetrahedron of type $0$ or $3$ whether the neighbor along a face $f$ of the root is a pyramid.
    \label{Alg:tetPyra}}
\end{algorithm}

\ourSubSubsection{Coordinates and type of the neighbor}
Next, we solve the second problem, i.e., computing the type and the anchor coordinate of the neighbor based on its shape. As above, the computation is separated into pyramidal and tetrahedral input elements.

In \tabref{pyra-neigh-type},
the types and shapes of a pyramid's neighbors are given for each face,
that is, tetrahedra of type $0$ or $3$ for the triangular faces $f_0$, $f_1$, $f_2$, $f_3$ and the opposite pyramid for the quadrilateral $f_4$.
To compute the coordinates of the neighbor, we consider two cases
which can be organized in one table each:
\enumerateCases{
	\item
    For a pyramid, the differences of the neighbor's anchor coordinates
    compared to the pyramid itself are given in
    \tabref{pyra-neigh-coord}.
    It shows that for tetrahedral neighbors, only the $x$- and $y$-coordinates differ, while the $z$-coordinate is the same;
    for the pyramidal neighbor at face $f_4$, in contrast, only the $z$-coordinate changes.
	\item For a tetrahedron with a pyramid as neighbor, the computation of the coordinates can be done in a similar fashion, see \tabref{tet-neigh}.
}

\begin{SCtable}
	\begin{tabular}{ccc}
		\toprule
		Face & pyramid $6$ & pyramid $7$ \\
		\midrule
		$f_0$ & $3$ & 3\\
		$f_1$ & $3$ & 3\\
		$f_2$ & $0$ & 0\\
		$f_3$ & $0$ & 0\\
		$f_4$ & $7$ & 6\\
		\bottomrule
	\end{tabular}%
    \caption{The type of the neighbor of a pyramid given a face of the
pyramid and its type.%
\tablab{pyra-neigh-type}}
\end{SCtable}%
\begin{table}
	\begin{tabular}{cccccc}
		\toprule
                Coordinate
		& $f_0$ & $f_1$ & $f_2$ & $f_3$ & $f_4$ \\
		\midrule
                $x$      & $0$ & $\plen ~||~ 0$ & $0$ & $0  ~||~ -\plen$ & $0$ \\
                $y$      & $0$ & $0  ~||~-\plen$ & $0$ & $\plen  ~||~ 0$ & $0$ \\
                $z$      & $0$ & $0$ & $0$ & $0$ & $-\plen  ~||~ \plen$ \\
		\bottomrule
	\end{tabular}%
        \caption{The manipulation of the coordinate of a pyramid to compute the
coordinates of the neighbor along a given face, with $\plen=2^{\mathcal{L}-l}$ the length of the face given its level $l$.  For each coordinate, the entry
in the table has to be added to the coordinate of the input pyramid.
We abbreviate pairs of assignment: $a ~||~ b$ means that $a$ is chosen if
$\ptype = 6$ and $b$ otherwise, i.e., meaning $\ptype = 7$.}
\tablab{pyra-neigh-coord}%
\end{table}%

\begin{table}
	\begin{tabular}{ccccc}
		\toprule
                Coordinate
		& $f_0$ & $f_1$ & $f_2$ & $f_3$ \\
		\midrule
		$x$ & $(\plen ~||~ 0)$ & $0$ & $0$ & $(0 ~||~ -\plen)$ \\
		$y$ & $(0 ~||~ \plen)$ & $0$ & $0$ & $(-\plen ~||~ 0)$ \\
		$z$ & $0$ & $0$ & $0$ & $0$\\
		\midrule
		Type & $7$ & $7$ & $6$ & $6$ \\
		\bottomrule
	\end{tabular}%
        \caption{The manipulation of the coordinate of a tetrahedron
neighboring a pyramid along a given face. For each coordinate the entry at the
table has to be added to the coordinate of the input
pyramid.
By $(a~||~b)$, we abbreviate the cases for type $0$ ($a$) or the other types ($b$).}%
\tablab{tet-neigh}%
\end{table}%

\ourSubSubsection{Dual face number}
Eventually, we determine the number of the dual face. If the neighbor of a tetrahedron is a tetrahedron too, we use the already existing algorithm for tetrahedra. If the input element is a pyramid, the dual face number depends on the pyramid type and the given face number. Given  the type of the input element and the face, the computation of the dual face can be done efficiently using \tabref{dual-face-num}.
\begin{SCtable}
	\begin{tabular}{cccccc}
		\toprule
                Type
		& $f_0$ & $f_1$ & $f_2$ & $f_3$ & $f_4$ \\
		\midrule
		$0$ & $f_3$ & $f_2$ & $f_2$ & $f_3$ & $-$ \\
		$3$ & $f_1$ & $f_0$ & $f_0$ & $f_1$ & $-$ \\
		$6$ & $f_2$ & $f_3$ & $f_2$ & $f_3$ & $f_4$ \\
		$7$ & $f_1$ & $f_0$ & $f_1$ & $f_0$ & $f_4$ \\
		\bottomrule
	\end{tabular}%
        \caption{The (dual) face number of the neighbor of an element given the input
type and a face.}%
\tablab{dual-face-num}%
\end{SCtable}%

If we run the costly check whether a tetrahedron has a pyramid neighbor last, we can avoid it as often as possible. All other checks run in constant time.

\subsection{Inter-tree neighbor elements}
\seclab{outside-root-neighbor}

This section addresses the computation of neighbors outside the root
pyramid. We follow the strategy used in \tetcode to compute neighbors over
the boundaries of a tree. It is based on the fact that two elements share a
face planar with the tree boundary.
This computation is split in two main parts:
\enumerateSteps{
\item Within the root pyramid, we have to be able to collapse an element to the face that touches the root's boundary.
\item On the other side, i.e., in the adjacent coarse-mesh element, we need the inverse of this process: Given a face on the boundary,
we want to construct the associated element in the root pyramid.
}

Note that other elements in the forest may have different maximal levels, which has to be considered when computing the neighbor over the face of the root element. For the sake of simplicity, however, we neglect this additional computation here, as it only affects the scaling of the computed coordinates.

\ourSubSubsection{Collapsing to a face}


In the following, we show how to collapse an element down to a given face,
starting with a pyramid:
For each face element, \tabref{pyra_to_boundary} lists its anchor coordinates and its type.
While the anchor coordinate of the quadrilateral face
matches that of the pyramid,
those of the four triangular faces are obtained through a coordinate permutation. The coordinates need to be scaled to the correct maximal level.

Note that (unlike prisms, pyramids, tetrahedra, and triangles)
quadrilateral elements do not have a type.
\begin{table}
	\begin{tabular}{cccccc}
		\toprule
                   Coordinate
		   & $f_0$ & $f_1$ & $f_2$ & $f_3$ & $f_4$ \\
		\midrule
		$T.x$ & $y$  & $y$ & $x$ & $x$ & $x$\\
		$T.y$ & $z$ & $z$ & $z$ & $z$ & $y$ \\
		\midrule
		Type & 0& 0&0 &0 & quad \\
		\bottomrule
	\end{tabular}
        \caption{The coordinates and type of a face of pyramid $P$ that touches
the root element with a given face number of $P$. For the faces $0$ to $3$ the
faces are triangles of type $0$; for face $4$, the face is a quadrilateral element, which does not have a type. The
coordinates $x$, $y$, and $z$ refer to the coordinates of the input element.}%
\tablab{pyra_to_boundary}%
\end{table}%
\begin{table}
	\begin{tabular}{ccccc}
		\toprule
        (Type, face) & $(0, f_1)$, $(2, f_2)$ & $(0, f_0)$, $(1, f_0)$ &
              $(3, f_1)$, $(1, f_2)$ & else \\
		\midrule
		$T.x$ & $y$ & $y$ & $x$ & $x$ \\
		$T.y$ & $z$ & $z$ & $z$ & $z$  \\
		\midrule
		Type of face & $(0 \pquest $1$ ~||~ 0)$ & $(0 \pquest $1$ ~||~ 0)$ &
                       $(3 \pquest $1$ ~||~ 0)$ & $(3 \pquest $1$ ~||~ 0)$ \\
		\bottomrule
	\end{tabular}
        \caption{The coordinates and type of a tetrahedron's face that touches
the root element. Both depend on the face and type of the element. All boundary
elements are triangles. We abbreviate the choice of type in the last row by $(t \pquest a ~||~ b)$, where
$a$ applies to type $t$ and $b$ otherwise.}%
\tablab{tet_to_boundary}%
\end{table}%
To instead collapse a tetrahedron
within the root pyramid,
we first have a look at the following example:
Given a pyramid that is refined once, we see that its tetrahedral children are between
two pyramids, and have type $0$ or $3$. Since the boundary surfaces of the pyramidal children
have type $0$, the boundary surfaces of a tetrahedral child have type $1$. If we refine
this tetrahedron again, another type of tetrahedron arises in
the middle. If we refine further, we observe that the type of the tetrahedra in the middle of the face alternate between two different types. Depending on the given face, the new arising type is $2$ or
$1$. These elements touch the root with a triangular face of type $0$.

In conclusion, if
the tetrahedron touching the root face has type $0$ or $3$, the triangle has type $1$, otherwise it has
type $0$, as listed in \tabref{tet_to_boundary}.

\ourSubSubsection{Construct pyramid from face}
Next, we discuss the inverse operation:
construct a volume element that touches a given surface element on a face of the root pyramid.
Since tetrahedra have already been introduced in previous work, this section is restricted to constructing pyramidal elements.

Determining the pyramid's anchor coordinate differs from face to face;
\tabref{boundary_to_pyra} provides an overview.
\begin{itemize}
    \item From a quadrilateral face element we can directly compute the pyramid. The anchor coordinates are given by those of the quadrilateral extended by a $z$-coordinate of zero; in the end, we scale the coordinates to the length of the root pyramid.
    \item For all triangular faces,
the $z$-coordinate of the pyramid is given by the $y$-coordinate of the face;
setting its $x$- and $y$-coordinates
is only slightly more involved:
as detailed in \tabref{boundary_to_pyra},
they are either given directly from the face's $x$- and $y$-coordinates
or
have to be computed as
the difference in length
between the root element and a pyramid of the same level as the triangle.
\end{itemize}



\begin{SCfigure}
	\resizebox{40mm}{!}{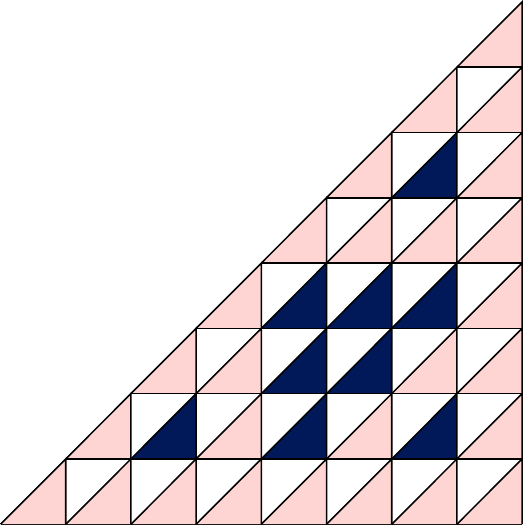}
	\caption{A triangular face of a pyramid that has been refined $3$ times. All white triangles have type $1$ and we know directly that they are extruded to a tetrahedron. We marked all type $0$ triangles that expand to a pyramid in light pink. All blue type $0$ triangles extrude to a tetrahedron.\label{fig:tri_to_pyra}}
\end{SCfigure}
In addition to the anchor coordinates, we need to compute the type of the pyramid.
We know for all triangles of type $1$ that the extruded element is a tetrahedron of type $0$ or type $3$ (depending on the given root face).
Accordingly,
only
if the triangle has type $0$, we need to check whether the extruded element is a pyramid or a tetrahedron, see Figure \ref{fig:tri_to_pyra} for an visualization of the extrusion of different triangles.
The element that a type-$0$ face extrudes to is a tetrahedron
if and only
if it has a tetrahedron as a parent; otherwise, it is a pyramid. In Theorem \ref{Thm:tri_to_pyra}, we show how to check the property on a binary level.
\begin{theorem}
        Let $T$ be a triangle of type $0$ on the surface of the root pyramid. The
boundary element with $T$ as a face is a pyramid if and only if
\begin{equation}
  T \rightarrow y = T \rightarrow x \bitwand T \rightarrow y ,
  \label{Thm:tri_to_pyra}
\end{equation}
 with ``$\bitwand$'' referring to a bitwise \textnormal{AND}.
\end{theorem}
\begin{proof}[Proof of Theorem \ref{Thm:tri_to_pyra}]
	Let $T$ be a triangle of type $0$ on the surface of the root pyramid
    that touches a pyramid as an inter-tree neighbor.
    As pyramids touching a face of its parent only occur in the corners, all the ancestors of $T$ lie in a corner of their parent. The statement holds for triangles of level $0$ and we use an inductive argument to show the first direction:

	Assume it holds for the first $l$ bits that $T\rightarrow y = T\rightarrow x \,\& \, T\rightarrow y$. In the corners of a triangle of type $0$ are the children with local id $0$, $2$, and $3$.
    \begin{itemize}
    \item
    The child with local id $0$ does not change the coordinates.
    \item
    The child with local id $2$ flips the ($l+1$)-th bit to $1$ for the $x$-coordinate, while the $y$-coordinate remains unchanged. Thus the equation still holds.
    \item
    For the child with local id $3$, both the $x$- and the $y$-coordinate are flipped to $1$ at the $(l+1)$-th bit and the equation still holds.
    \end{itemize}
	Next, let T be a triangle of type $0$ which fulfills the equation. Then every bit of the coordinates has to fulfill the equation. Each bit describes the position of a triangle in its parent. The combinations of bits of $x$- and $y$-coordinates fulfilling the equation are $(0,0), (1,0)$, and $(1,1)$. The type-$1$ child of a triangle fulfills the equation too, but its type-$0$ child at the next level does not. Thus, at every bit the coordinates refer to a corner triangle. Therefore, the extruded element with face $T$ has the shape of a pyramid.
\end{proof}

The type of the neighboring tetrahedron depends on the type of the triangle and the number of the root face and can be computed via a look up \tabref{boundary_to_tet_type}.
At last, we need to compute the number of the face of the pyramid touching the root face. If the shape is a pyramid, it is the number of the given root face. Otherwise it depends again on the type of the triangular element and the given root face, and can be computed via a look up \tabref{face_num_root}.

\begin{table}
	\begin{tabular}{cccccc}
		\toprule
                Coordinate
		& $f_0$ & $f_1$ & $f_2$ & $f_3$ & $f_4$ \\
		\midrule
		$P.x$ & $y$ & $2^{L} - 2^{T\rightarrow\text{level}}$ & $x$& $x$ & $x$\\
		$P.y$ & $x$ & $x$& $y$& $2^{L} - 2^{T\rightarrow\text{level}}$ & $y$\\
		$P.z$ & $y$ & $y$& $y$& $y$ & 0\\
		\bottomrule
	\end{tabular}
        \caption{The coordinates of an extruded element on the face of the
root pyramid. We notice that on face $4$ the elements are quadrilaterals, the
elements on all other faces are triangles. We refer with $x$ and $y$ to the
coordinates of the triangle $T$ on the face of the root element.}%
\tablab{boundary_to_pyra}%
\end{table}%


\begin{table}
\begin{subtable}{.48\textwidth}
\centering
\begin{tabular}{ccccc}
        \toprule
         Type
         & $f_0$ & $f_1$ & $f_2$ & $f_3$ \\
        \midrule
        $0$ & $1$ & $1$ & $2$ & $0$ \\
        $1$ & $0$ & $0$ & $3$ & $3$ \\
        \bottomrule
    \end{tabular}%
    \caption{The type of the tetrahedron
    touching the root pyramid, given a root face.}
    \tablab{boundary_to_tet_type}
\end{subtable}
\begin{subtable}{.48\textwidth}
\centering
    \begin{tabular}{ccccc}
						\toprule
                                                   Type
						   & $f_0$ & $f_1$ & $f_2$ & $f_3$ \\
						\midrule
						$0$ & $2$ & $0$ & $2$ & $0$ \\
						$1$ & $1$ & $0$ & $1$ & $0$ \\
						\bottomrule
\end{tabular}%
\caption{The type of the constructed tetrahedron touching the root face.}%
\tablab{face_num_root}
\end{subtable}
\caption{Relations between the faces of the root pyramid and tetrahedra.}
\end{table}

\section{Algorithms at the forest level}
\seclab{highlevel}

Global, high-level procedures for mesh manipulation, information gathering,
or searching can be implemented independently of the element-specific
routines.
The forest-level view allows general processing and communication of mesh
metadata by simply calling the necessary element-specific virtual functions.
In \secref{Sec:resultcube} and \secref{Sec:resultplane}, we will investigate
how the pyramid-specific functions influence the performance and scalability
of a set of exemplary forest-level algorithms:
\begin{itemize}
    \item \tetnew creates a uniformly refined forest of a given depth that
is partitioned equally over all available processes.
    \item \tetadapt evaluates a user-provided criterion to decide whether an element should be refined, stay as it is, or whether a family of elements should be coarsened into their parent element.
    \item \tetpartition repartitions a forest; based on the number of
elements in the mesh, the boundaries for each process are computed and the
elements (and possible data connected to them) are communicated to
reestablish an equal-weight partition over all processes.
    \item \tetghost determines the remote face-neighbor elements of every
process and gathers them in a symmetric point-to-point communication scheme.
\end{itemize}

When generating an initial forest, the first task is to compute a
partitioned uniform mesh via \tetnew. This involves determining the assignment of coarse
mesh elements (trees) based on a target refinement level, ensuring a
balanced element count per process after the initial uniform refinement.

For standard isotropic refinement, the number of elements in a uniformly
refined tree at level $\ell$ is determined by the refinement factor (e.g.,
$N(\ell) = 2^{\text{dim}\cdot\ell}$).
If all element types share the same refinement ratio, partitioning is
straightforward, requiring only the total element count.

However, the pyramid refinement introduced here transcends this simple
assumption.
Pyramidal descendants use different refinement factors (eight or ten),
depending on whether they are a pyramid or tetrahedron, leading to the
following element total for an $\ell$-times refined root pyramid,
\begin{equation}
  N(\ell) = 2 \cdot 8^\ell - 6^\ell ,
\end{equation}
see \ref{Number of pyramids} for a quick derivation.
Furthermore,
mixing pyramids with other element shapes means not all trees refine in the
same manner. Consequently, achieving a fair initial load assignment requires
a non-trivial algorithm, which we detail in \secref{forestnew}.

%

Considering future developments in adaptive refinement, a generalized methodology extending beyond the pyramidal case is beneficial for the \tetnew algorithm. This generality will support advanced techniques, such as the mixing of anisotropic and isotropic refinement, e.g., to refine prismatic elements within boundary layers preferentially along a single axis.

We rely on the following properties to compute the distribution of elements and trees of a forest:

\begin{enumerate}
  \item Each process maintains a consecutive range of local trees.
        Its first and last local tree may be shared with other processes.
        The range may also be empty for a process, in which case its first local tree is the first local tree of the next non-empty process; in this case, the last local tree id is $\texttt{first\_local\_tree\_id} - 1$, as it has no
        valid first or last local tree.
        A process stores the shape of its local trees' root elements.
  \item We maintain a global replicated array $\fK[0, P]$ that contains the
        cumulative number of trees before any process $p \in [0, P)$, thus
        $\fK_0 = 0$, $\fK_P = K$.
  \item In order not to count any tree twice in $\fK$, we define one
        responsible process $p_k$ for each tree $k \in [0, K-1]$.
        It is the lowest-id process that shares tree $k$.
        Empty processes can never be responsible for any tree
        (a difference to \cite{Burstedde20d}).
\end{enumerate}

\subsection{Forest partitioning with multiple shapes}
\seclab{forestnew}

The construction of a new forest takes a coarse mesh and an
initial uniform refinement level $\ell$ as parameters.
The goal of this algorithm is to build an equally partitioned forest of
elements of level $\ell$ exclusively.
On input,
the coarse mesh may be distributed in an arbitrary way.
Since every process only creates elements of the forest in its
local trees, we conclude that the input coarse mesh is in general not
partitioned to allow for a fair distribution of elements.
Thus, the first part of the forest construction is to compute the boundaries of each process, which are then used to fill the forest with its leaf elements.

For each shape of root elements, we can compute the number of elements in the associated tree on level $\ell$ in $O(1)$, i.e., constant runtime.
This number, however, might vary for each tree.
Therefore, we let each process responsible for one or more trees compute the
number of elements $N(\ell)$ each tree will hold at level $l$ and store it for
later reference.
This local number of elements is communicated to all processes via an
\texttt{MPI\_Scan}, from which each process derives a local copy of the
non-descending array $\fC[0, P]$;
for each process $p$, it holds the current element offset, i.e.,
the global index of the first element within the trees process $p$ is responsible for.
By construction, its final entry is the total number of elements:
$\fC_0 = 0$, $\fC_P = N$.
%

After repartitioning,
the $N$ elements shall be equally distributed among all $P$ processes.
In the resulting partitioning, we can determine the element range of process $p$ by the familiar formula
\begin{equation}
  \eqnlab{partitioncut}%
  O_p = \floors{\frac{p N}P}, \qquad N_p = O_{p + 1} - O_p,
\end{equation}
where $O_0, \ldots, O_P$ are the ideal element offsets
and the value of $N_p$ differs by at most one among all processes.
The other way around,
the element with global index $E$
will belong to the process
\begin{equation}
	p(E) = P - 1 - \floors{\frac{P(N - 1 - E)}{N}} .
	\eqnlab{elementtoproc}
\end{equation}

To achieve an equal distribution, the current tree- and element boundaries have to be shifted.
To perform this shift, we determine the range of processes each process needs to send information about its trees to.
Based on $\fC$ we can compute the global indices $T$ of the first elements of each tree of process $p$ locally.
We use \eqnref{elementtoproc} to compute the processes $0 \le q_0 \le q_1 < P$ which need the boundary information of  $n_0 = T[0]$ and $n_1 = T[K_p]$, where $K_p$ is the number of trees on process $p$. We use a binary search to find the processes $q_0$ and $q_1$.
The processes in the range $[q_0, q_1]$ need information from process $p$ about its boundaries. Each of them will thus receive at least one message from process $p$.


The overall algorithm takes the following steps:
\begin{enumerate}
	\item Compute the process-local number of elements.
	\item Use a prefix reduction to produce an offset array $\fC$.
	\item Produce an array $T$, where entry $T_i$ is the global index of the first element of local tree $i$.
	\item For each local tree $i$, compute the smallest process number requiring information about its elements (via a binary search in $T$).
    \item Each process computes the up to two processes it retrieves information from.
	\item Send tree- and element boundaries to each process for which we have information.
\end{enumerate}
The details of the algorithm are presented in Algorithm \ref{Alg:uniform_bounds}.
The performance of the algorithm, among others, is discussed in
\secref{Sec:resultcube} and \secref{Sec:resultplane}.




\begin{algorithm}
	\TitleOfAlgo{\texttt{uniform\_bounds}}
        \KwData{The calling process $p$ computes the boundary information of
a uniform partition based on the tree and elements it owns
and sends the information to processes that need the information.
It also receives this information.}
    \setlength{\algoCommentLength}{5.5cm}
   \DontPrintSemicolon
    \small
    \myAlgoLine{Global index of the first element of proc $p$}
    {index = \texttt{global\_index}$(\texttt{first\_tree}(p))$}
    \myAlgoLine{Array of offsets, use \texttt{MPI\_Scan}}
    {$\fC[p]$ = index}
    \For{\textnormal{$i=0,\ldots,K_p$}} {
    \myAlgoLine{Array of global ids of first tree elements}
    {$T[i] = \texttt{first\_element}(\texttt{tree}(i))$}
    }
    \myAlgoLine{Use \eqnref{elementtoproc}, the first process to send to}
    {$q_0 = \texttt{proc\_to\_send}(T[0])$}
    \myAlgoLine{Use \eqnref{elementtoproc}, the last process to send to}
    {$q_1 = \texttt{proc\_to\_send}(T[K_p])$}
    \myAlgoLine{At q, the index of its $1^{st}$ tree in T}
    {proc\_to\_tree = $\texttt{bsearch}(T, q_0, q_1)$}
    \For{$q = q_0, \ldots, q_1$}
    {
        \myAlgoBlockHeader{Compute the first tree $FT$ and last tree $LT$ to send to $q$}

        \myAlgoLine{First and last element of $q$, use \eqnref{partitioncut}}
        {$FE_q$, $LE_q$, to be filled. }

        \If{\texttt{is\_empty}($q$)}
        {
            \myAlgoLine{$FT_q$ will be sent to $q$}
            {$FT_q = FT_{\text{next-non-empty-proc}}$}
        }
        \myAlgoLine{Possible first tree of $q$}
        {first\_tree = proc\_to\_tree[$q$]}
        \myAlgoLine{Possible first element of $q$}
        {$FE$ = $T[\mathrm{first\_tree}]$}
        \uIf{$FE \geq T[K_p]$}{
            \myAlgoLine{Proc $p$ has no information for $q$}
            {skip to the next process}
        }
        \Else{
            \myAlgoBlockHeader{Check if the calling process can send the lower or upper bound to $q$}
            \uIf{$FE$ > $FE_q$}
            {
                \myAlgoLine{The lower-bound info is on another proc}
                {continue}
            }
            \Else{
                \myAlgoLine{Shift by -1 if $FE$ is on shared tree}
                {$FT_q = \texttt{shift}(\mathrm{first\_tree}$)}
                \myAlgoLine{Tree-local id of the first element of $q$}
                {$FE_q = O_q - T[FT_q]$}
            }
            \uIf{$T(N) < LE_q$}{
                \myAlgoLine{The upper bound info is on another proc}
                {continue}
            }
            \Else{
                \myAlgoLine{Set the last tree of $q$}
                {$LT_q = \mathrm{proc\_to\_tree}[q + 1] - 1$}
                \myAlgoLine{Tree-local id of the last element of $q$}
                {$LE_q = O_{q+1} - T[LT_q] - 1$}
            }

            }
        \myAlgoLine{Non-blocking send}
        {Send $FT_q$ and $LT_q$ to $q$}

    }
    \myAlgoLine{Search procs that send to calling proc}
    {$q_{\mathrm{recv\_first}}, q_{\mathrm{recv\_last}} = \texttt{binary\_search}(\fC)$}

    \myAlgoLine{Receive boundary information}
    {$\texttt{recv}(q_{\mathrm{recv\_first}}, q_{\mathrm{recv\_last}})$}

    \caption{Compute uniform bounds for an initial forest partition.}%
    \label{Alg:uniform_bounds}%
\end{algorithm}

\section{Element performance}
\seclab{Sec:resultcube}

In this set of tests we compare the performance of our pyramid
implementation with the existing implementations of space-filling curves for
standard elements in three dimensions.

To ensure a fair comparison between the different types of elements, we make
sure that the number of elements created is roughly the same.
For all types we create a uniform mesh consisting of several trees. It is
then refined once and partitioned, before a ghost layer is created. The
following setups are used for the tests:
\begin{itemize}
	\item For pyramids we refine all elements of type $0$, $2$, $4$, and $6$. The initial mesh has eight trees and we start with a uniform mesh on level $7$, $8$, $9$, or $10$.
	\item For all other element types we refine every second element. The initial mesh has two trees and the uniform mesh has level $8$, $9$, $10$, or $11$.
\end{itemize}

We run three tests, each doubling the number of processes. That way, we can compare weak and strong scaling for each type of element. The results of this test series are summarized and visualized in Figures \ref{fig:Adapt_performance}, \ref{fig:Partition_Performance}, and \ref{fig:New_Performance}. For the smallest examples we create $143\times 10 ^6$ (pyramids) or $150\times 10^6$ (other) elements in total, resulting in at least $2\times 10^6$ elements per process.
The largest benchmark creates $75\times 10^9$ or $77\times 10^9$ elements,
respectively.

All experiments have been executed on CARA, a high-performance cluster run by the German Aerospace Center (DLR). It has $2168$ nodes with two AMD EPYC 7601 ($32$ cores, $2.2$ GHz) each. Each node has access to $128$ GB DDR4 ($2666$ MHz) RAM. An Infiniband HDR network connects the nodes \cite{CARA}.

For all core algorithms of \tetcode we observe a similar scaling behavior
between pyramids and the other element shapes.
Especially for the \tetpartition and the \tetadapt algorithms we observe
nearly ideal strong and weak scaling.

\begin{figure}
    \centering
    \includegraphics[scale=0.5,trim={1.5cm 0 2.5cm 0},clip]{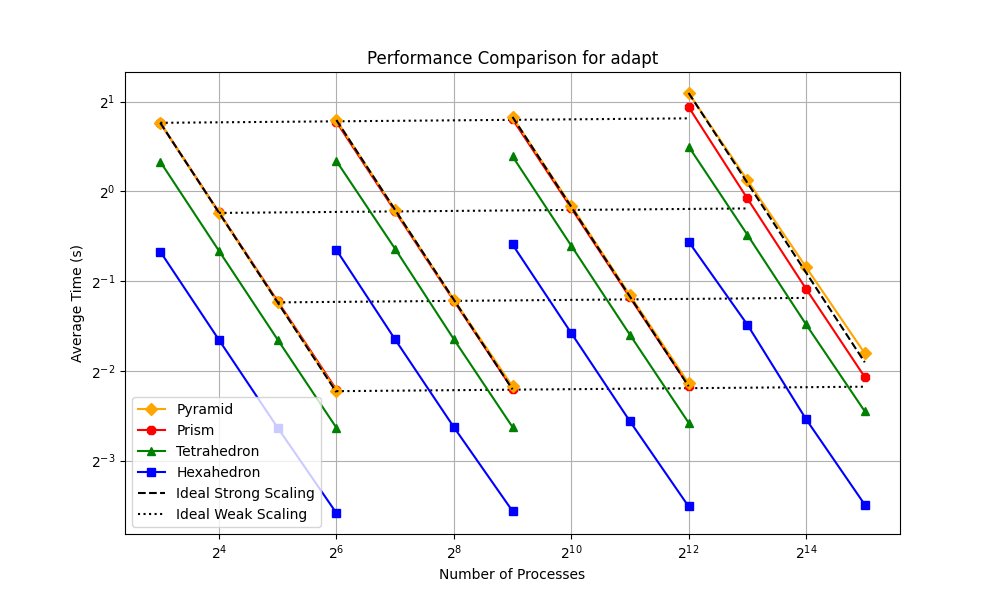}
    \caption{\label{fig:Adapt_performance}The runtimes of the \tetadapt algorithm
executed for four different problem sizes.
We demonstrate near ideal strong and weak scaling for all types of elements.}
\end{figure}

\begin{figure}
	\includegraphics[scale=0.5,trim={1.5cm 0 2.5cm 0},clip]{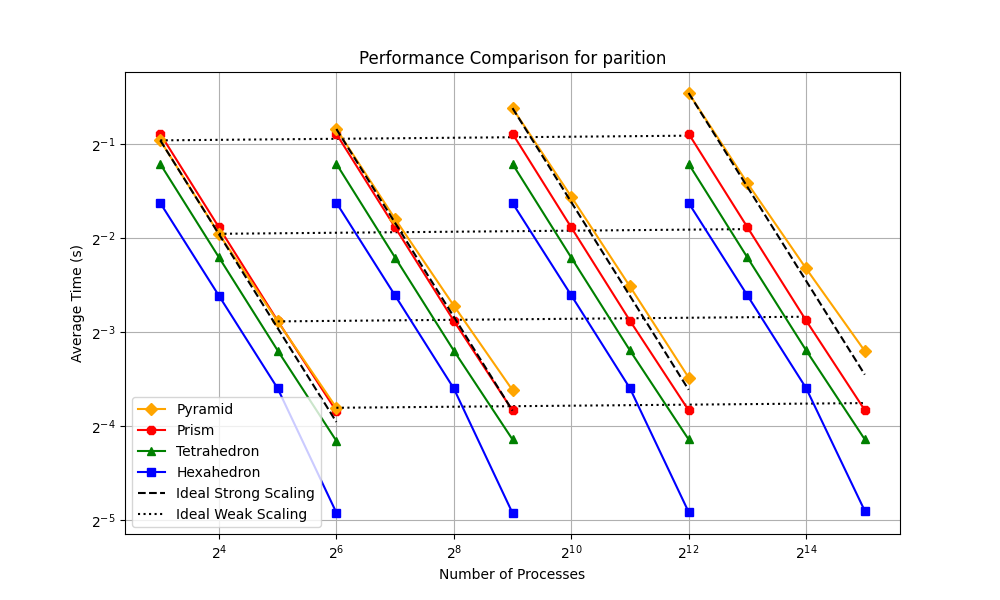}
	
        \caption{An evaluation of strong and weak scaling of the
\tetpartition algorithm of \tetcode for various element shapes.
The configuration is noted below in Figure \ref{fig:New_Performance}.%
}%
\label{fig:Partition_Performance}
\end{figure}

For the \tetnew algorithm the strong scaling is nearly ideal too; however,
we observe a small jump in the runtime for the two largest configurations
for the pyramidal algorithms.
We suspect to have reached a near worst-case scenario for the boundaries in
the pyramidal case.
It results in a more costly communication pattern, increasing the runtime.

\begin{figure}
    \centering
	\includegraphics[scale=0.5,trim={1.5cm 0 2.5cm 0},clip]{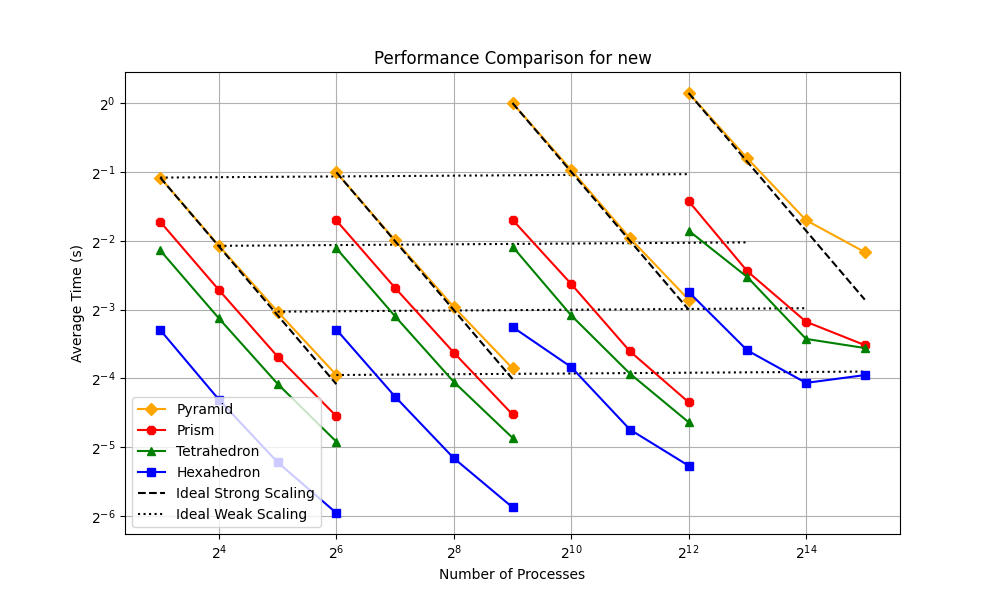}	
        \caption{A performance evaluation of
the \tetnew algorithm of \tetcode (Algorithm~\ref{Alg:uniform_bounds}).
We create a forest with a small number of trees and prescribe a uniform
refinement of at least level 7.
For each level we double the number of processes.
After doubling the number of processes three times we increase the level of
the uniform forest once.
The ideal scaling for pyramids is not constant, because the number of
pyramids in a mesh increases by a factor of greater than 8 on average
when uniformly refining.}%
\label{fig:New_Performance}%
\end{figure}


In general, the algorithms for pyramids are a bit more costly, which is
caused by the increased complexity of the calculations at the element level.
However, for forest-level algorithms that do not prompt weighty per-element
computation, the runtime is similar to other non-hexahedral elements.
Hexahedra run the fastest, because the element-level algorithms do not have
to consider different element types or type-dependent branching.
In fact, the hexahedral element routines are the only code linked in from
\pforest, wrapped in \tetcode's virtual function interface, but we would
expect a native \tetcode implementation to be of similar speed.
All in all, we preserve the near ideal scaling behavior of \tetcode for all
types of elements, including the pyramids introduced in this work.

The runtime analysis of the \tetghost algorithm \cite{HolkeKnappBurstedde19}
is complex:
The determination of ghost elements involves a depth-first search,
leading to an added logarithmic factor in terms of the elements per process.
On the other hand, the size of the ghost layer itself, which represents a
surface area, grows with an exponent of $2/3$ compared to the element count
of a process.
As the focus of this work is not the \tetghost algorithm, we omit a detailed
runtime discussion. Nonetheless, we present the runtime at a glance in \ref{tab:Table_Ghost_efficiency} and \ref{tab:Table_weak_efficiency}, demonstrating that the
inclusion of pyramidal elements performs comparably to existing elements and
does not introduce additional computational overhead within an application.


\section{Hybrid mesh example}
\seclab{Sec:resultplane}

In this section, we use a hybrid mesh to demonstrate the performance and
scalability of Algorithm \ref{Alg:uniform_bounds}, the \tetnew algorithm
introduced in \secref{forestnew}.
The coarse mesh is depicted in \figref{tonne_mesh}; it consists of around
100,000 trees, containing all 3D element shapes supported by \tetcode.
Initially, the mesh is refined five times.
During adaptation, all elements inside a virtual skewed wall are refined
twice.
After partitioning and the computation of the ghost layer, the wall is moved
and we adaptively refine again.
In total, the wall is moved five times and we construct a refined mesh
consisting of up to $40 \times 10^9$ elements.

\begin{figure}
	\includegraphics[scale=0.15]{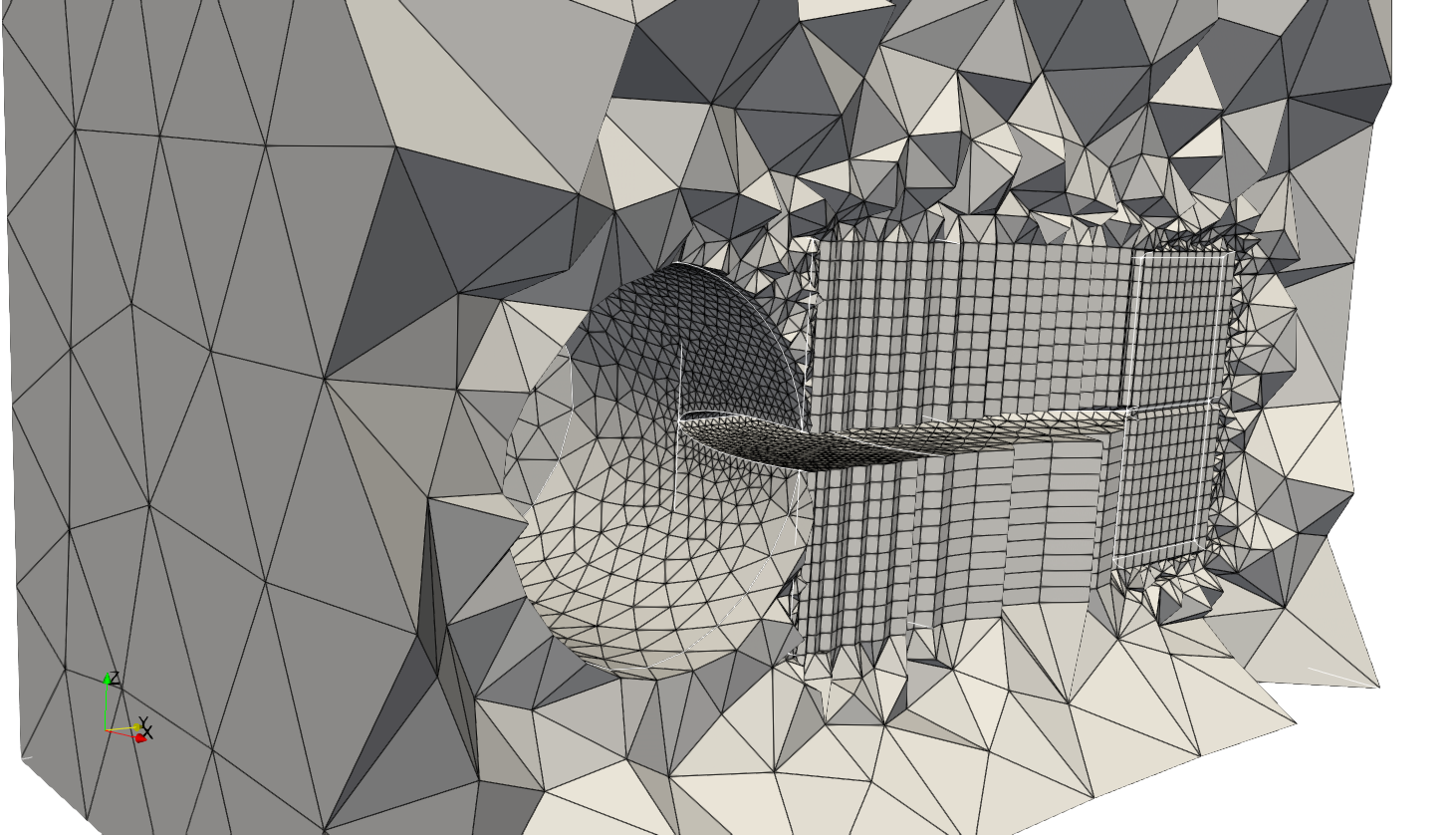}
        \caption{The mesh used for the benchmarks in \secref{Sec:resultplane}.
It is a toy example of an airplane geometry with a single wing.
The mesh consists of 69,431 tetrahedra, 3,800 hexahedra, 29,520 prisms and
3,120 pyramids.
The mesh can be found in \cite{tetdata25}.
The elements near the ``foil'' are extruded prisms, the far field consist of
tetrahedra.
At the boundary, hexahedra and pyramids occur, resulting in a suitable
hybrid CFD mesh%
.}%
\label{fig:tonne_mesh}%
\end{figure}

The test has been performed on CARO \cite{CARO}, a high performance cluster
run by the German Aerspace Center.
It is the sister system of CARA, which was in downtime and therefore not
available for this test. CARO consists of 1,276 nodes with two AMD EPYC
7702 (64 cores, $2.0$ GHz) each.
Each node we used has access to 256 GB DDR4 (3200 MHz) RAM. The nodes
are connected via an Infiniband HDR network.

\begin{table}
    \begin{tabular}{ccccc}
        \toprule
         Processes
         &  $2^{14}$ & $2^{15}$ & $2^{16}$ & $2^{17}$\\
        \midrule
        Total runtime / s   & $250.262$& $126.602$ & $69.943$ &   $49.942$ \\
        Adapt runtime / s & $224.099$ & $111.728$ & $57.835$ & $36.083$   \\
        Adapt / Total & $89.54\%$ & $88.65\%$ & $83.81\%$  & $72.27\%$   \\
        \bottomrule
    \end{tabular}
        \caption{The total runtime of the hybrid
mesh example. Between $72\,\%$ and $90\,\%$ of this runtime is spent in
\tetadapt. The runtime of the algorithm depends on the cost of the
adaptation criterion. In comparison to the criterion used in
\secref{Sec:resultcube}, this geometric criterion to refine in the
vicinity of a moving virtual wall is more costly.}%
\label{tab:hybrid_mesh}%
\end{table}

Table~\ref{tab:hybrid_mesh} illustrates the results of the experiments: The
total runtime of each test is shown for up to 131,072 processes. Most of the
time is used to evaluate the adaptation criterion, which is an expensive one
in this case. Each element has to test whether it is inside the wall, or
not.

The runtime of the other core algorithms of \tetcode are negligible in
comparison, because they do not perform any geometry evaluation.
Here it suffices to gather information about the elements in the reference
space.


For the first three runs, each process contains more than a million
elements, and a near ideal strong scaling can be observed. For the last run
($2^{17}$ processes), the scaling is slightly less than ideal, since the
number of elements drops down to roughly 300,000 elements per process and
the communication overhead becomes noticeable.

\section{Conclusion}
\seclab{conclusion}

In this work, we have successfully introduced and formalized a novel
space-filling curve for pyramidal elements.
We demonstrated, through a rigorous reference implementation, how to
integrate pyramids into tree-based, hybrid adaptive mesh-refinement
software.
This advancement is significant because it enables pyramids to function
effectively as a bridge element within scalable tree-based hybrid meshes,
significantly increasing the geometric flexibility and robustness of meshes
used for complex simulations.
%
%
Furthermore, we introduce a new algorithm to compute uniform bounds
for the initial partitioning of a hybrid mesh.

We executed extensive tests to demonstrate the scalability and efficiency of
our design.
The core algorithms of \tetcode maintain the same favorable performance
characteristics for pyramid elements as they do for the existing hexahedral,
tetrahedral, and prismatic elements.
We further demonstrate this capability through the parallel efficiency of
\tetcode on a large-scale hybrid mesh (over 100,000 trees), utilizing up to
131,072 processes.
In summary, this work establishes an efficient and scalable foundation for
hybrid adaptive mesh refinement (AMR).

In future works, we plan to simplify the existing Morton-type SFC implementations
across all standard element shapes, through a unified theory that can
potentially yield further performance improvements.
Second, we intend to integrate this hybrid tree-based AMR infrastructure
into common numerical solver libraries to provide a tangible improvement
with respect to indexing, data locality, and parallel efficiency.

\section*{Acknowledgements}

We acknowledge funding by the DFG under grant no.\ 467255783 (``Hybride
AMR-Simulationen'').
We acknowledge travel funds by the Hausdorff Center for Mathematics (HCM) at
the University of Bonn.
The HCM is funded by the German Research Foundation (DFG) under Germany’s
excellence initiative EXC 59 -- 241002279 (``Mathematics: Foundations,
Models, Applications'').

The author Lukas Dreyer has been supported by the Deutsche Forschungsgemeinschaft (DFG) under Germany`s Excellence Strategy within the Cluster of Excellence PhoenixD (EXC 2122, Project ID 390833453).

We chose equal-weighted scientific color maps \cite{Crameri23}, and employed
Gemini Thinking 3 Pro to polish the general style of the paper.

\bibliographystyle{siamplain}
\bibliography{bibtex/ccgo_new, bibtex/carsten}
\bibliography{refs}

\ifarXiv
    \foreach \x in {1,...,\numbersupplementpages}
    {
        \clearpage
        \includepdf[pages={\x,{}}]{\supplementfilename}
    }
\fi


\end{document}